\documentclass[prx,superscriptaddress,aps, twocolumn,
10pt,longbibliography, floatfix]{revtex4-2}

\usepackage{amsmath,amssymb,amsthm,mathrsfs,amsfonts,dsfont,mathtools,physics}
\usepackage{braket}

\usepackage{ragged2e}
\usepackage{graphicx}

\usepackage{algorithm}       
\usepackage{algpseudocode}

\usepackage{hyperref}
\usepackage[capitalize]{cleveref}

\usepackage{float}

\usepackage{microtype}

\usepackage[caption=false]{subfig}
\usepackage{soul}
\usepackage{adjustbox}

\usepackage{bm}
\usepackage{enumerate}
\usepackage{color}

\usepackage{appendix}

\usepackage{enumitem}
\usepackage[normalem]{ulem}
\usepackage{comment}
\usepackage{xcolor}

\usepackage{amsthm}

\newtheorem{remark}{Remark}
\newtheorem{definition}{Definition}

\newtheorem{theorem}{Theorem}
\newtheorem{lemma}{Lemma}
\newtheorem{proposition}{Proposition}
\newtheorem{corollary}{Corollary}

\Crefname{theorem}{Theorem}{Theorems}
\theoremstyle{remark}

\newcommand{\qmaddress}{\affiliation{Quantum Motion, 9 Sterling Way, London N7 9HJ, United Kingdom}}
\newcommand{\oxaddress}{\affiliation{Mathematical Institute, University of Oxford, Woodstock Road, Oxford OX2 6GG, United Kingdom}}
\newcommand{\mrnaaddress}{\affiliation{Moderna, Cambridge, MA 02139, USA}}

\begin{document}

\title{Nearly Optimal Amplitude Estimation at any Depth}

\author{Jona Erle}
\email{jona.erle@balliol.ox.ac.uk}
\oxaddress
\qmaddress
\mrnaaddress

\author{B\'alint Koczor}
\email{balint.koczor@maths.ox.ac.uk}
\oxaddress
\qmaddress

\begin{abstract}
We develop a class of amplitude estimation algorithms with tunable circuit depth $M$ and circuit repetitions $N$, requiring neither ancilla qubits nor controlled Grover operations. For additive error $\epsilon$ in the Grover angle, they attain the nearly optimal query--depth tradeoff $M^2N\in\tilde{\mathcal{O}}(\epsilon^{-2})$ uniformly over $\lambda\in[0,\pi/2]$, spanning the full range from classical sampling at $M=1$ to the Heisenberg limit at $M=\Theta(\epsilon^{-1})$. While prior depth-tunable work establishes comparable angle-accuracy guarantees only away from the boundaries $\lambda\to0$ or $\pi/2$ or at discrete depth tradeoffs, our guarantee extends to both boundaries, so the quantum speedup persists there rather than degrading to classical sampling. Numerical experiments confirm the predicted uniform angle accuracy and show low overhead in practice, making them strong candidates for practical amplitude estimation in the early fault-tolerant regime, with applications such as overlap certification, trial-state verification, and Monte Carlo methods. 
\end{abstract}

\maketitle

\emph{Introduction.} 
Amplitude estimation (AE) is a ubiquitous quantum algorithmic primitive and a fundamental subroutine of many fault-tolerant quantum algorithms. Given a unitary $U$ preparing the state $\ket{\psi}=U\ket{0}$ and an orthogonal projector $P$ onto a ``good'' subspace, the task is to estimate the amplitude $a:= \bra{\psi}P\ket{\psi}\in[0,1]$. In practice, this is often done by estimating the angle $\lambda\in [0,\pi/2]$ defined through $a=\sin^2(\lambda)$. Throughout, we quantify error in the Grover angle: $\epsilon$ denotes a target additive error in $\lambda$. An estimate $\hat{\theta}$ with $|\hat{\theta}-\lambda|\leq \epsilon$ yields $\hat{a}=\sin^2(\hat{\theta})$ satisfying
\begin{equation*}
    |\hat{a}-a|\leq 2\sqrt{a(1-a)}\epsilon+\epsilon^2.
\end{equation*}
Uniform accuracy in $\lambda$ is strictly stronger than uniform additive accuracy in $a$, but is also the natural benchmark for a quantum speedup---classical sampling achieves $\mathcal{O}(N^{-1/2})$ accuracy in $\lambda$ uniformly over $[0,\pi/2]$. Thus, retaining a quantum speedup at the boundaries naturally calls for a uniform accuracy guarantee in $\lambda$.
Classical sampling requires $\Theta(\epsilon^{-2})$ queries to $U$ to estimate $\lambda$ uniformly to accuracy $\epsilon$, whereas AE requires only $\Theta(\epsilon^{-1})$ queries to $U$ and its inverse $U^\dagger$, yielding a quadratic improvement that attains the Heisenberg limit and is optimal up to constant factors \cite{nayak1999quantum, brassard2000quantum}. This scaling underpins quadratic quantum speedups in Monte Carlo methods \cite{montanaro2015quantum} and algorithms with applications in areas such as finance \cite{rebentrost2018quantum, woerner2019quantum, stamatopoulos2020option,  chakrabarti2021threshold}, quantum chemistry \cite{knill2007optimal, huang2026fullqubit}, and optimization \cite{sidford2023quantum, van2020quantum}.

Textbook AE~\cite{brassard2000quantum} uses controlled powers of the Grover operator, an $\mathcal{O}(\log\epsilon^{-1})$-qubit ancilla register, and an inverse quantum Fourier transform (QFT). All three ingredients---the coherent control, the ancilla overhead, and the QFT---are prohibitive in the early fault-tolerant (EFT) regime \cite{katabarwa2024early,zimboras2025myths}, in which the number of logical qubits and the achievable coherent circuit depth are severely constrained. Algorithms targeting the EFT regime therefore trade coherent operations for circuit repetitions and classical post-processing \cite{lin2022heisenberg, wan2022randomized, dong2022ground, ding2023even}. Consequently, considerable effort has gone into developing AE algorithms that remove both the QFT and the controlled Grover operations, replacing them with classical inference on measurement outcomes of Grover circuits of varying depth \cite{aaronson2020quantum, suzuki2020amplitude, nakaji2020faster, grinko2021iterative, labib2024quantum}.

These simplifications alone, however, do not render AE feasible on EFT hardware as a maximal circuit depth $M \in \Theta(\epsilon^{-1})$ is still required. Recent work has therefore developed low-depth AE
algorithms~\cite{giurgica2022low, rall2023amplitude, vu2025low, huang2026low, sun2026quantum} which trade maximal circuit depth $M$ for an increased number of circuit repetitions $N$, thereby operating between the standard quantum limit ($M = \Theta(1)$, $N = \Theta(\epsilon^{-2})$) and the Heisenberg limit ($M = \Theta(\epsilon^{-1})$, $N = \tilde{\mathcal{O}}(1)$). Any AE algorithm attaining accuracy $\epsilon$ must obey the Zalka--Burchard lower bound as  $M^2N\in \Omega(\epsilon^{-2})$~\cite{zalka1999grover, burchard2019lower}, and while several match it up to polylogarithmic factors for $\lambda$ bounded away from $0$ and $\pi/2$, none does so uniformly over $\lambda\in[0,\pi/2]$. For example, Power-law AE, introduced by Giurgica-Tiron \textit{et al.}~\cite{giurgica2022low}, relies on the regularity conditions of the Bernstein--von Mises theorem, which fail at the boundary, i.e. at $\lambda\to0,\pi/2$. Rall and Fuller's hybrid quantum-classical amplitude estimation \cite{rall2023amplitude} requires amplitude-dependent circuit depths and query overheads, both of which diverge as $\lambda\to 0$. Huang and Koczor's Gaussian Least Squares Amplitude Estimation (GLSAE) \cite{huang2026low} falls back to classical sampling at the boundary. We refer the reader to \cref{app:related_work} for a more detailed discussion of prior work.

In this work, we propose Windowed Least Squares Amplitude Estimation (WLSAE), a family of algorithms that sample Grover circuit depths from a discrete probability distribution, which in the following we refer to as a window, and estimate $\lambda$ through a simple one-dimensional least-squares fit. WLSAE generalizes GLSAE from Gaussian windows to an explicitly characterized class of admissible windows, and we prove that every admissible window achieves the query-depth tradeoff $M^2N\in \tilde{\mathcal{O}}(\epsilon^{-2})$ uniformly for every $\lambda\in [0,\pi/2]$, thereby matching the Zalka--Burchard bound up to polylogarithmic factors across the full interpolation from classical sampling to the Heisenberg limit. Since GLSAE is the Gaussian member of this class, it follows in particular that the boundary patch of~\cite{huang2026low} is unnecessary, and, in the regime where it is invoked, forgoes the quantum speedup entirely. Having a family of windows rather than a single window further allows us to optimize window design for practical implementations. We numerically compare admissible windows against each other, and show that a simple uniform window already outperforms the Gaussian window of GLSAE.

\emph{Measurement Model.} As before, let $U$ be a unitary preparing the state $\ket{\psi}=U\ket{0}$, let $P$ be an orthogonal projector, and let 
\begin{equation*}
    G:=-U(I-2\ket{0}\bra 0)U^\dagger (I-2P)
\end{equation*}
denote the Grover operator. Let $a:=\bra{\psi}P\ket{\psi}$ and $\lambda\in[0,\pi/2]$ be defined through $a=\sin^2(\lambda)$. Ref.~\cite{huang2026low} showed that for every $m\in \mathbb{Z}$ there exists an ancilla-free circuit that makes $|m|$ queries to $U$ and $U^\dagger$ and returns an outcome $Z_m\in\{-1,1\}$ with 
\begin{equation}
    \label{eq:signal}
    \mathbb{E}[Z_{m} \mid m] = \cos(2 m \lambda),\quad \mathrm{Var}[Z_m\mid m]=\sin^2(2m\lambda).
\end{equation}
In particular, for even $|m|=2t$, we prepare the state $(I-2P)G^{t-1}\ket{\psi}$ and measure the shifted Loschmidt echo $2\ket{\psi}\bra{\psi}-I$, whereas for odd $|m|=2t+1$, we prepare the state $G^t\ket{\psi}$ and measure $I-2P$. For $m=0$ no circuit is executed and \cref{eq:signal} gives $Z_0=1$ deterministically. Notably, this measurement model requires neither controlled Grover operations nor ancilla qubits. 

\emph{Windowed Least Squares Amplitude Estimation.}
A window of maximal depth $M$ is a discrete probability distribution supported on $W_M:=\{-M,\dots,M\}$ (see \cref{fig:1} for examples). Given a window $p$, WLSAE draws $N$ i.i.d.\ depths $m_1,\dots,m_N\sim p$, executes the corresponding circuits of \cref{eq:signal}, and records the outcomes $Z_{1}, \dots, Z_{N} \in \{\pm 1\}$. Since depth $|m|$ costs $|m|$ queries, the total query count is $Q=N\mathbb{E}_{m\sim p}[|m|]\leq NM$ in expectation. The estimate $\hat{\theta}$ of $\lambda$ is then obtained by minimizing the mean-square loss function
\begin{equation}
    \label{eq:loss}
    L(\theta):=\frac{1}{N}\sum_{i=1}^N\left(Z_{i}-\cos(2m_i\theta)\right)^2
\end{equation}
over the uniform grid 
\begin{equation}
    G_K:=\{(j+\frac{1}{2})\frac{\pi}{2K}:0\leq j< K\}\subset [0,\pi/2],
\end{equation} 
where the grid size $K\in\mathbb{N}_{\geq 1}$ is an input to the algorithm. The corresponding amplitude estimate is $\hat{a}=\sin^2(\hat{\theta})$. Taking $p$ to be a Gaussian window recovers GLSAE \cite{huang2026low}. The full procedure is given in \cref{alg:ldae}. 
\begin{figure}[htbp]
    \begin{minipage}{1\linewidth}
        \begin{algorithm}[H]
  \caption{WLSAE}
  \label{alg:ldae}
  \begin{algorithmic}[1]
    \Require Unitary $U$, projector $P$, available depth $M$, number of shots $N$, window $p$ on $W_M$, grid size $K$
    \State $G\gets -U(I-2\ket{0}\bra{0})U^\dagger(I-2P)$, $\ket{\psi}\gets U\ket{0}$
    \For{$i = 1, \dots, N$}
    \State $m_i\sim p$ \Comment{Sample depth from window}
    \If{$m_i=0$}
    \State $Z_i\gets 1$\Comment{no circuit executed}
    \ElsIf{$|m_i| = 2t$, $t \in \mathbb{N}$}
    \State $Z_i\gets$ measure $2\ket{\psi}\bra{\psi}-I$ on $(I-2P)G^{t-1}\ket{\psi}$
    \ElsIf{$|m_i| = 2t+1$, $t \in \mathbb{Z}_{\geq 0}$}
    \State $Z_i\gets$ measure $I-2P$ on $G^{t}\ket{\psi}$
    \EndIf
    \EndFor
    \State $G_K\gets \{(j+\frac{1}{2})\frac{\pi}{2K}:0\leq j< K\}$
    \State $L_{\min}\gets \infty$, $\hat{\theta}\gets 0$
    \ForAll{$\theta \in G_K$}
    \State $L(\theta)\gets\frac{1}{N}\sum_{i=1}^{N}\left(Z_{i}-\cos(2m_i\theta) \right)^2$
    \If{$L(\theta)<L_\mathrm{min}$}
    \State $L_\mathrm{min}\gets L(\theta)$, $\hat{\theta}\gets \theta$
    \EndIf
    \EndFor
    \State $\hat{a}\gets \sin^2(\hat{\theta})$
    \State \Return $\hat{\theta}$, $\hat{a}$
  \end{algorithmic}
\end{algorithm}
    \end{minipage}

\end{figure}

\begin{figure*}[htbp]
  \includegraphics{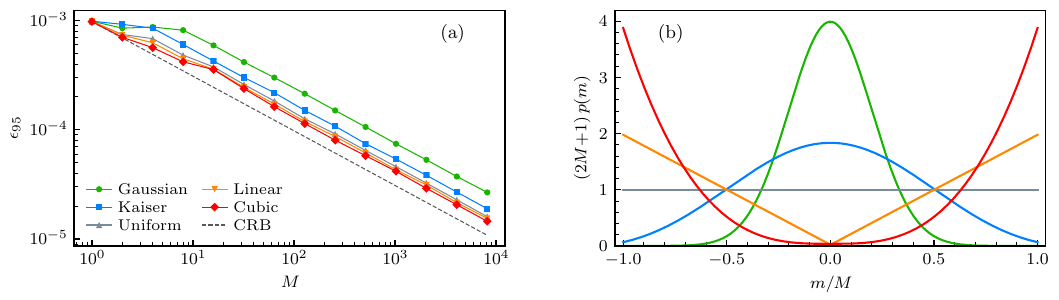}
  \vspace{-2em}
  \caption{WLSAE for five window functions at fixed query budget $Q=10^6$ and fixed $\lambda=0.5$, over $10^4$ independent runs with grid size $K=\lceil 100\sqrt{\mathbb{E}[m^2]N}\rceil $. (a) $\epsilon_{95}$ against maximal circuit depth $M$, where the dashed line is the Cram\'er--Rao bound (CRB). (b) The five windows plotted as $(2M+1)p(m)$ against $m/M$ for $M=10^3$. The (unnormalized) windows are defined as:  
  Gaussian $\exp(-m^2/(2T^2))$ with $T=\max(1,M/5)$; Kaiser $I_0(\beta\sqrt{1-(m/M)^2})$ with $\beta=5$; uniform $1$; linear $f+|m|/M$, and cubic $f+(|m|/M)^3$ with $f=0.01$.
  }
  \label{fig:1}
\end{figure*}

\emph{Main Results.} \Cref{alg:ldae} is well defined for any window, but the choice of the window determines whether it achieves the nearly optimal query-depth tradeoff. Two extremes illustrate what is required. A window supported on a constant small depth, e.g., $p(1)=1$, yields a well-behaved estimator, but its precision scales as $\mathcal{O}(N^{-1/2})$ irrespective of $M$, so no quantum speedup is obtained. A window supported on a single large depth, say $p(M)=1$, produces a loss that is sharply curved at $\theta=\lambda$ but has $\Theta(M)$ minima on $[0,\pi/2]$, so that $\lambda$ cannot be uniquely identified. An appropriate window must therefore place enough mass at large depths to resolve $\lambda$ with high precision, and enough mass across the remaining depths to distinguish the true minimum from its aliases. In \cref{def:admissible} we introduce a class of windows that provably achieves both, and therefore attains the nearly optimal query-depth tradeoff.
\begin{definition}[Admissible window]
    \label{def:admissible}
    Let $M\in\mathbb{N}_{\geq 1}$, let $\sigma\geq 1$ and $b\in(0,1]$ be constants independent of $M$, and set $T=\max(1, M/\sigma)$. A window $p$ on $W_M$ is called $(b,\sigma)$-admissible if 
    \begin{equation}
        \label{eq:admissible_cond}
        p(m)+p(-m)\geq \frac{b}{T},
    \end{equation}
    for all integers $m$ with $1\leq m \leq T$.
\end{definition}
\Cref{eq:admissible_cond} requires the window $p$ to place mass at least $b/T$ on every depth up to $T$. Hence, a constant fraction of its total mass is below $T$ and $T$ can be seen as the window's effective width, while $\sigma$ bounds how far beyond $T$ the window's tail may reach. The truncation at $T\geq 1$ simply keeps \cref{eq:admissible_cond} satisfiable in the regime where $M<\sigma$.

For instance, the uniform window on $W_M$ is $(\frac{2}{3}, 1)$-admissible while the truncated Gaussian of width $T$ used by GLSAE \cite{huang2026low} is $(0.48, 5)$-admissible. Notably, $(b, \sigma)$-admissibility is sufficient, but not necessary, for the following guarantees to hold.
\begin{theorem}
    \label{theorem:main}
    Let $p$ be a $(b,\sigma)$-admissible window on $W_M$, $\epsilon\in (0,1/M]$, and $\delta\in(0,1)$. There exists a constant $A\geq1$, depending only on $b$ and $\sigma$, and a choice of $N$ with $M^2N\in{\mathcal{O}}\left(\epsilon^{-2}\log(\epsilon^{-1}\delta^{-1})\right)$, where the implied constant depends only on $b$ and $\sigma$, such that, upon choosing $K=\lceil \frac{A\pi}{2\epsilon}\rceil$, \cref{alg:ldae} returns, for every $\lambda\in[0,\pi/2]$, with probability at least $1-\delta$, an estimate $\hat{\theta}$ of $\lambda$ satisfying $|\hat{\theta}-\lambda|\leq \epsilon$. Consequently, $\hat{a}=\sin^2(\hat{\theta})$ satisfies \begin{equation}
        \label{eq:error_prop}
        |\hat{a}-a|\leq 2\sqrt{a(1-a)}\epsilon +\epsilon^2.
    \end{equation}
\end{theorem}
In particular, at the endpoints $a= 0$ or $1$, \cref{eq:error_prop} reduces to $|\hat{a}-a|\leq \epsilon^2$. Thus, \cref{theorem:main} provides strictly more information than a uniform additive-$\epsilon$ guarantee in $a$.
The following corollary follows as an immediate consequence.

\begin{corollary}
    \label{cor:interpolation}
    Let $\beta\in[0,1]$, $\epsilon\in(0,1]$, and $\delta\in(0,1)$. There exists an algorithm that outputs, with probability at least $1-\delta$, an estimate $\hat{a}$ of $a$ satisfying \cref{eq:error_prop} for every $a\in[0,1]$, using ${\mathcal{O}}(\epsilon^{-1-\beta}\log(\epsilon^{-1}\delta^{-1}))$ queries to $U$ and $U^\dagger$ at maximum circuit depth $M\in\mathcal{O}(\epsilon^{-1+\beta})$, without ancilla qubits or controlled Grover operations.
\end{corollary}
\Cref{cor:interpolation} follows the convention of \cite{huang2026low}, such that we attain the Heisenberg limit at $\beta=0$ and recover classical sampling at $\beta=1$. In \cref{theorem:main}, and hence at every $\beta$ in \cref{cor:interpolation}, the query--depth product satisfies $M^2N\in\mathcal{O}(\epsilon^{-2}\log(\epsilon^{-1}\delta^{-1}))$, matching the Zalka--Burchard bound $M^2N\in\Omega(\epsilon^{-2})$ \cite{zalka1999grover, burchard2019lower} up to a single logarithmic factor. Hence WLSAE, applied to admissible windows, attains the nearly-optimal query--depth tradeoff and, unlike previous low-depth algorithms, does so for all $\beta\in[0,1]$ and all $a\in[0,1]$.

While full proofs are deferred to \cref{app:proof}, we briefly explain why \cref{alg:ldae} succeeds across all amplitudes, even though---as observed in~\cite{huang2026low}---the loss flattens at the boundary. 
In particular, \Cref{cor:envelope} states
\begin{equation}
    C_2\Psi^2\leq \mathbb{E}\left[L(\theta) - L(\lambda)\right]\leq C_1\Psi^2,
\end{equation}
where $C_1,C_2\in O(1)$, $\Psi:=\min (1, T\Delta, T^2\Delta\max(\bar{\lambda}, \Delta))$
with $\Delta := |\theta - \lambda|$,
$\bar{\lambda} := \min(\lambda, \frac{\pi}{2} - \lambda)$, and, by \cref{def:admissible}, $T=\max(1,M/\sigma)$. Consequently, around the minimum at $\lambda$, i.e. whenever $\Delta \leq 1/T$, the loss grows quadratically in $\Delta$ away from the boundary, whereas as $\bar\lambda\to 0$ it grows only quartically, so that a fixed deviation in $L$ certifies only a much larger deviation in $\theta$. The noise, however, vanishes at the same rate. Let $L(\theta)-L(\lambda)=\frac{1}{N}\sum_{i=1}^NY_i$ with $Y_{i} := (Z_{i} - \cos(2m_{i}\theta))^{2} - (Z_{i} - \cos(2m_{i}\lambda))^{2}$. From \cref{remark:variance}, it follows that, for $\Delta\leq 1/T$,
\begin{equation}
    \mathrm{Var}[Y_i]\leq C_6\frac{\Psi^4}{T^2\Delta^2}.
\end{equation}
Therefore, for $\Delta\leq 1/T$, we find
\begin{equation}
    \label{eq:snr}
     \frac{\mathbb{E}\left[L(\theta) - L(\lambda)\right]}{\sqrt{\mathrm{Var}[Y_{i}]/N}} \geq \frac{C_{2}}{\sqrt{C_{6}}}\, T\Delta\sqrt{N},
\end{equation}
in which $\Psi$---and with it the entire $\lambda$-dependence---has canceled, giving a $\lambda$-independent lower bound on the signal-to-noise ratio. This cancellation is invisible to concentration inequalities that bound deviations by the range of their summands rather than by their variance, because the range does not vanish fast enough: By \cref{remark:variance}, $\max Y_{i}-\min Y_i \leq C_{7}\Psi$ exceeds the standard deviation by a factor $\max(1, 1/(T\max(\bar{\lambda},\Delta)))$. A Hoeffding-type argument therefore requires $\Psi \gtrsim \sqrt{r}$ rather than $T\Delta \gtrsim \sqrt{r}$, where $r = \log(2K/\delta)/N$. The two agree whenever $T\bar{\lambda} \gtrsim 1$---precisely the regime in which the guarantee of \cite{huang2026low} is established---but at the boundary the former yields only $\Delta \gtrsim r^{1/4}/T$, degrading the rate in $N$ from $N^{-1/2}$ to $N^{-1/4}$. 

The loss depends on the sampled data only through the per-depth counts $n_m=|\{i:|m_i|=m\}|$ and outcome sums $S_m=\sum_{i:|m_i|=m}Z_i$. We can write
\begin{equation}
    \label{eq:loss_statistics}
    L(\theta)=\frac{3}{2}+\frac{1}{N}\sum_{m=0}^M\left[\frac{n_m}{2}\cos(4k\theta)-2S_k\cos(2k\theta) \right],
\end{equation}
which is a discrete cosine polynomial and whose values on $G_K$ are returned by a single discrete cosine transform. Therefore, the classical cost of \cref{alg:ldae} is $\mathcal{O}(N+K\log K)$ compared to the naive $\mathcal{O}(NK)$. 

\emph{Numerics.}
\Cref{theorem:main} guarantees that WLSAE achieves a nearly-optimal query--depth tradeoff for every admissible window. However, in practice, constant overheads are equally important and---especially in the early fault-tolerant regime---may decide whether an algorithm is feasible or not. We therefore evaluate WLSAE numerically, comparing window functions and probing the amplitude-uniformity that \cref{theorem:main} promises. Throughout, $Q=\mathbb{E}_{m\sim p}[|m|]N\leq MN$ denotes the expected number of queries made to $U$ and $U^\dagger$, and we report $\epsilon_{95}$, the $95$th percentile of $|\hat{\theta}-\lambda|$ over independent runs, i.e. the accuracy $\epsilon$ certified at confidence $1-\delta=0.95$. 

\Cref{fig:1} compares five windows at fixed query budget $Q=10^6$ and shows that tapering towards large $|m|$ is counterproductive: at $M=1024$ the uniform window attains an error approximately $40\%$ below that of the Gaussian window of GLSAE \cite{huang2026low}, and $17\%$ below that of the Kaiser window, while the anti-tapered linear and cubic windows improve on the uniform window by a further $5\%$ and $10\%$. This is the opposite of what one would expect from spectral estimation \cite{harris1978use, wang2023quantum, ding2024quantum}, where tapering suppresses side lobes at the cost of a wider main lobe. Side lobes, however, are not what limit AE: the signal contains only a single frequency, so there is no leakage from neighboring frequencies to be suppressed. What the window must instead deliver is a steeply rising $\mathbb{E}[L(\theta)-L(\lambda)]$ around its minimum $\lambda$, together with enough distinct depth scales to separate that minimum from its aliases. Tapering therefore sacrifices the former to buy a suppression that AE does not need.

How much precision a window buys per query is bounded by the Fisher information. A shot at depth $|m|$ carries $\mathcal{I}(m)=4m^2$, independent of $\lambda$, so at fixed budget $Q$ a window $p$ on $W_M$ collects
\begin{equation}
    \mathcal{I}(Q)=\frac{4Q\mathbb{E}_{m\sim p}[m^2]}{\mathbb{E}_{m\sim p}[|m|]},
\end{equation}
which grows with the mass placed at large depths and is maximized by $p(M)+p(-M)=1$. That window, however, is neither admissible nor able to distinguish the true minimum from its aliases, so $\mathcal{I}(Q)$ cannot be the only design objective, but it accounts for the ordering in \Cref{fig:1}(a), and admissibility is precisely the constraint that stops one from maximizing it. The dashed line in \Cref{fig:1}(a) is $\epsilon_{\mathrm{CRB}}:=1.96\sqrt{1/\mathcal{I}(Q)}$, evaluated for the window that maximizes the Fisher information over $W_M$, and converted to a $95$th percentile for Gaussian errors. Since the Cram\'er--Rao bound (CRB) assumes regular unbiased estimation and is blind to the global alias structure, it is not a rigorous lower bound for our least-squares estimator, but serves as a useful local benchmark away from the boundary, where \cref{fig:1} is evaluated. At large $M$ the observed $\epsilon_{95}$ of the uniform window lies within a factor of approximately $3/2$ of this benchmark. 

Furthermore, at fixed query budget $Q$ the shot count $N=Q/\mathbb{E}_{m\sim p}[|m|]$ is fixed by the window alone. At $M=1024$ and relative to the uniform window, the Gaussian and Kaiser windows require $3.1$ and  $1.5$ times as many shots, whereas the linear and cubic windows require only $0.75$ and $0.63$ times as many. Since every shot carries an initialization and readout cost independent of $|m|$, this may result in a further overhead in total runtime. 

\begin{figure}[htbp]
  \includegraphics{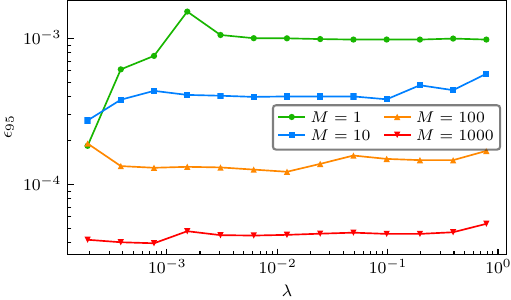}
  \caption{Amplitude-uniformity of WLSAE applied to uniform window at fixed query budget $Q=10^6$, over $10^4$ independent runs with grid size $K=\lceil 100\sqrt{\mathbb{E}[m^2]N}\rceil $. $\epsilon_{95}$ against $\lambda$ for four maximal circuit depths $M$. Without loss of generality, only the endpoint $\lambda\to0$ is shown. \Cref{alg:ldae} is exactly equivariant under $\lambda\mapsto\pi/2-\lambda$, so $\epsilon_{95}$ is symmetric about $\lambda=\pi/4$ and the $\lambda\to\pi/2$ endpoint requires no separate sweep.}
  \label{fig:2}
\end{figure}

\Cref{fig:2} sweeps $\lambda$ over three orders of magnitude at fixed query budget $Q=10^6$ using the uniform window. The curves do not degrade at the boundary and therefore confirm the amplitude-uniformity promised by \cref{theorem:main}---each simply shifts down by $M^{-1/2}$ as the depth increases. The dips of the $M=1$ and $M=10$ curves at small $\lambda$ arise because the estimate is squeezed against the boundary of the search interval. Since $\hat{\theta}\geq 0$ while $\lambda$ is close to $0$ the error can no longer fluctuate symmetrically around $\lambda$, and the $95$th percentile shrinks accordingly. In the extreme $\lambda\lesssim 1/(M\sqrt{N})$ a typical run observes no $-1$ outcome at all, the loss $L(\theta)$ is minimized at $\theta=0$, and the estimator returns the smallest grid point $\hat{\theta}=\frac{\pi}{4K}$, so that $|\hat{\theta}-\lambda|\leq \max\left(\lambda, \frac{\pi}{4K} \right)$. The grids used here are much coarser than \cref{lemma:main} prescribes: we made no attempt to optimize the constants of \cref{app:proof}, where the dominant looseness sits in \cref{lemma:lower_bound_near}.

Overall, the numerics show that a simple uniform window already outperforms the Gaussian window of GLSAE, and that anti-tapered admissible windows improve on it further---while GLSAE has in turn been shown to be competitive with state-of-the-art low-depth techniques~\cite{huang2026low}. Therefore our results suggest that WLSAE combines the strongest available theoretical guarantees with competitive practical performance.

\emph{Discussion.}
The regime in which existing guarantees fail is narrow, but not rare in practice. Guarantees of prior depth-tunable algorithms degrade once $\bar{\lambda}\lesssim\sigma/M$, corresponding to $a\lesssim (\sigma/M)^2$ near $a=0$ and $1-a\lesssim (\sigma/M)^2$ near $a=1$---a vanishing fraction of $[0,1]$, but precisely the fraction into which many applications fall. The reason is that for small amplitudes the meaningful quantity is the relative, not the absolute, error. Certifying $|\hat{a}-a|/a\leq \rho$ through \cref{eq:error_prop} requires $\epsilon \leq \sqrt{a}\left(\sqrt{1-a+\rho}-\sqrt{1-a} \right)$, for which $\epsilon \leq \sqrt{a}(\sqrt{1+\rho}-1)$ is sufficient at every $a\in[0,1]$ and tight as $a\to 0$. For fixed relative precision $\rho$, the tolerance therefore shrinks as $\epsilon \in \Theta(a^{1/2})$. A depth-tunable algorithm operating at depth $M\approx \epsilon^{-1+\beta}$ of \cref{cor:interpolation} is thus driven to $M\in \Theta(a^{(-1+\beta)/2})$, so that $a<a^{1-\beta}\approx M^{-2}$ for every $\beta>0$, up to $\rho$- and $\sigma$-dependent constants.
Consequently, any depth-limited algorithm estimating a small amplitude to fixed relative accuracy operates in the boundary region, and does so more deeply the shallower its circuits, i.e. the larger $\beta$. Only at the Heisenberg endpoint $\beta=0$ does the depth keep pace, which is why the boundary is invisible to full-depth amplitude estimation. In this region the fallback of \cite{huang2026low} spends $\Theta(\epsilon^{-2}\log\delta^{-1})$ queries, whereas \cref{cor:interpolation} spends $\mathcal{O}(\epsilon^{-1-\beta}\log(\epsilon^{-1}\delta^{-1}))$, a saving of $\epsilon^{-1+\beta}$ up to logarithmic factors. Equivalently, WLSAE spends $\tilde{\mathcal{O}}(a^{-(1+\beta)/2})$ queries against the $\Theta(a^{-1})$ samples required classically, retaining the full quantum advantage its depth budget permits.

The applications that operate in this regime are among the most natural targets for amplitude estimation. Overlap certification, that is deciding whether a prepared state has any considerable component on a target eigenstate, or estimating a ground-state overlap that decays with system size, is precisely the $a\to 0$ regime and a crucial task in quantum chemistry. Estimating acceptance probabilities of block encodings, tail probabilities in Monte Carlo risk estimation, and approximate counting with sparse solution sets also fall into the same regime. Trial-state verification and fidelity estimation sit at the opposite endpoint, where the quantity of interest is the infidelity $1-a$ and the argument above applies with $a$ replaced by $1-a$ throughout. In each case one wants a prescribed relative accuracy from circuits no deeper than the hardware allows, which is exactly what \cref{theorem:main} provides and where guarantees restricted to $\bar\lambda\gtrsim \sigma/M$ fail.

Compared to prior low-depth algorithms (see \cref{app:related_work} for details), \cref{cor:interpolation} is the first guarantee that holds simultaneously at both endpoints of the depth interpolation and uniformly in $a$. Power-law AE \cite{giurgica2022low} inherits the regularity conditions of the Bernstein--von Mises theorem, which fail as $\bar\lambda\to 0$, the hybrid algorithm of Rall and Fuller carries an explicit $\tilde{\mathcal{O}}(a^{-1/2})$ in query overhead \cite{rall2023amplitude}, QoPrime \cite{giurgica2022low} reaches only a set of discrete tradeoffs, excluding both endpoints, i.e. $\beta=0,1$, and GLSAE \cite{huang2026low} is the Gaussian member of our family, whose boundary patch we showed to be an artifact of the proof strategy rather than a statistical necessity.

\emph{Conclusion.}
We introduced Windowed Least Squares Amplitude Estimation, a depth-tunable and ancilla-free amplitude estimation algorithm parametrized by a discrete probability window over Grover circuit depths, and proved that every admissible window achieves the query-depth tradeoff $M^2N\in\mathcal{O}(\epsilon^{-2}\log(\epsilon^{-1}\delta^{-1}))$ for every angle $\lambda\in[0,\pi/2]$ and every point $\beta\in[0,1]$ of the interpolation between classical sampling ($\beta=1$) and the Heisenberg limit ($\beta=0$), matching the Zalka--Burchard bound $M^2N\in \Omega(\epsilon^{-2})$ up to a single logarithmic factor. 

To our best knowledge, this is the first low-depth amplitude estimation algorithm to attain the nearly-optimal query--depth tradeoff uniformly in the angle $\lambda$ across the entire depth interpolation. Numerically, the constant overheads are also small---a simple uniform window attains roughly $40\%$ lower error than GLSAE at equal query budget and maximal depth. Since WLSAE requires no ancillas, no controlled Grover operations, and no coherence beyond whatever depth the device supports, every improvement in hardware translates directly into improved accuracy at fixed budget, which makes it both a candidate for practical amplitude estimation in the early fault-tolerant regime and a testbed for benchmarking the approach to quantum advantage.

\bigskip
\emph{AI Usage \& Data Availability.}
The code for the numerics was developed with the assistance of Anthropic's Claude Opus 5 and is available on  \url{https://github.com/ox-quant-info/WLSAE}. Furthermore, Claude Opus 5 assisted in formulating some passages of the text.

\bigskip
\emph{Acknowledgments.}
The authors thank Po--Wei (George) Huang  for valuable discussions.
This research was supported by a grant funded by Moderna and Quantum Motion for an industrial PhD Studentship.
B.K. thanks UKRI for the Future Leaders Fellowship Theory to Enable Practical Quantum Advantage (MR/Y015843/1).
This research was funded in part by UKRI (MR/Y015843/1).
For the purpose of Open Access, the author has applied a CC BY public copyright licence
to any Author Accepted Manuscript version arising from this submission.

\bibliography{refs_ae}

\onecolumngrid
\appendix
\crefalias{section}{appendix}
\section{Related Work}
\label{app:related_work}
In 2022 Giurgica-Tiron \textit{et al.} proposed the first low-depth amplitude estimation algorithms---Power-law AE and QoPrime AE \cite{giurgica2022low}. Their first algorithm, Power-law AE, can be seen as a low-depth version of the QFT-free amplitude estimation algorithm by Suzuki \textit{et al.} \cite{suzuki2020amplitude}. It combines a power-law depth schedule with maximum likelihood estimation, and for $\beta\in(0,1]$ is claimed to estimate the angle $\lambda$ to accuracy $\epsilon$ using $\mathcal{O}(\epsilon^{-1-\beta})$ queries at maximal depth $\mathcal{O}(\epsilon^{-1+\beta})$, provided that the regularity conditions required for the Bernstein--von Mises theorem hold. While the algorithm has been demonstrated experimentally \cite{giurgica2022ex}, the authors state these regularity conditions as a hypothesis rather than verify them. Their log-likelihood is given by
\begin{align}
    l(X,\theta) = 2\sum_k\left[N_{k_0}\log\cos((2k+1)\theta)+N_{k_1}\log\sin((2k+1)\theta) \right]+\log Z,
\end{align}
where $X$ denotes the measurement outcomes, $Z$ is a normalization constant independent of $\theta$, and $N_{k_0}$ and $N_{k_1}$ count the outcomes $0$ and $1$ among the $N_k$ shots taken after $k$ Grover operations. Their outcome $1$ occurs with probability $\sin^2((2k+1)\lambda)$, so with $m=2k+1$ it corresponds to $Z_m=-1$ in our convention, and $N_{k_0}=(n_m+S_m)/2$, $N_{k_1}=(n_m-S_m)/2$. They therefore use the same measurement statistics $(n_m, S_m)$ that we use in \cref{eq:loss_statistics}, restricted to odd $m$, but combine them logarithmically as
\begin{align}
    l(X, \theta) &= \sum_{m\text{ odd}}\left[(n_m+S_m)\log\cos(m\theta)+(n_m-S_m)\log\sin(m\theta) \right]+\log Z\\
    &=\sum_{m\text{ odd}}\left[n_m\log\left(\frac{1}{2}\sin(2m\theta)\right)+S_m\log\cot(m\theta) \right]+\log Z,
\end{align}
whereas WLSAE combines them quadratically. Let $\Theta\subseteq[0,\pi/2]$ be an open interval around the true value $\lambda$ and $\bar{\Theta}$ its closure. Among the conditions summarized in Appendix A of \cite{giurgica2022low} (following \cite{hipp1976bernstein}), the Bernstein--von Mises theorem requires the log-likelihood $l(X,\theta)$ to be continuous on $\bar{\Theta}$ for every possible measurement outcome, twice differentiable on $\Theta$, and to satisfy suitable local moment bounds on its second derivative. To obtain a guarantee uniform over all $\lambda \in[0,\pi/2]$, the parameter domain $\Theta$ must extend arbitrarily close to both boundaries. However, for any outcome with $n_m-S_m>0$ for some $m$, we have $l(X, \theta)\to -\infty$ as $\theta\to 0$ because $\log\sin(m\theta)$ diverges. Hence, the required continuity cannot hold on a domain whose closure contains $0$. Restricting to $\Theta=[\eta, \pi/2-\eta]$ restores continuity, but does not yield a uniform guarantee. In particular, their estimator cannot return any $\hat{\theta}\notin \Theta$, so for $\lambda<\eta$ its error is bounded below by $\epsilon\geq\eta-\lambda$. Hence, achieving accuracy $\epsilon$ for every $\lambda\in[0,\pi/2]$ requires $\eta\leq \epsilon$. 

Moreover, the constants entering the Bernstein--von Mises rate bound depend on the choice of $\Theta$ through local regularity quantities that diverge as $\Theta$ approaches either boundary. Therefore, there are two possibilities: either $\eta$ is fixed, in which case the analysis cannot provide a uniform guarantee for $\lambda<\eta$, or $\eta$ is allowed to shrink with the target accuracy $\epsilon$, in which case the constants in the BvM bound are no longer $\mathcal O(1)$ and the assumption $N_{\mathrm{shot}}\in\mathcal O(1)$ used in \cite{giurgica2022low} can no longer be justified uniformly. Thus their analysis establishes the claimed query--depth tradeoff only for amplitudes bounded away from the boundaries, but not uniformly across all $\lambda\in[0,\pi/2]$.

Their second contribution, QoPrime AE, is a number-theoretic algorithm with a fully rigorous proof based on the Chinese remainder theorem. While QoPrime comes with a uniform angle guarantee, it is limited to discrete parameter settings $\beta=q/k$ with $k\geq 2$ and $1\leq q\leq k-1$, and therefore cannot continuously interpolate between classical sampling at $\beta=1$ and the Heisenberg limit at $\beta=0$. Notably, both endpoints $\beta=0$ and $\beta=1$ are excluded, so QoPrime cannot attain the exact Heisenberg limit. Furthermore, numerical experiments in \cite{giurgica2022low} find QoPrime to perform worse than Power-law AE in practice.

\bigskip
In section 5 of \cite{rall2023amplitude}, Rall and Fuller propose a hybrid quantum-classical amplitude estimation algorithm within their quantum signal processing framework. It estimates the amplitude $a_\mathrm{RF}=\sin(\lambda)$---rather than $a=\sin^2(\lambda)$---by maintaining a confidence interval $[a_{\mathrm{min}}^{(t)}, a_{\mathrm{max}}^{(t)}]$, which contains $a_\mathrm{RF}$ with high probability, and shrinking it adaptively at every iteration $t$. At every iteration it constructs semi-Pellian polynomial transformations of $a_\mathrm{RF}$ using quantum signal processing (QSP) \cite{low2017optimal, gilyen2019quantum}. Since each polynomial depends on the current confidence interval and hence on the outcomes of all previous iterations their circuits cannot be fixed in advance and their algorithm is therefore not parallelizable. They state that for every $\epsilon, \delta>0$ and $0\leq \beta<1$ their algorithm finds an estimate $\hat{a}_{\mathrm{RF}}$ of $a_\mathrm{RF}$ satisfying $|\hat{a}_{\mathrm{RF}}-a_\mathrm{RF}|\leq\epsilon$ with probability at least $1-\delta$ using $\tilde{\mathcal{O}}((a_{\mathrm{RF}}^{-1}+\epsilon^{-(1+\beta)})\log\delta^{-1})$ queries at maximal circuit depth $\mathcal{O}(a_{\mathrm{RF}}^{-1/(1-\beta)}+\epsilon^{-(1-\beta)})$. 
An additive guarantee on $a_\mathrm{RF}$ is stronger than an additive guarantee on $a$ near $a\to 0$. In particular, given $\hat a_{\mathrm{RF}}$ with $|\hat a_{\mathrm{RF}}-a_{\mathrm{RF}}|\leq\epsilon$, the estimate $\hat a:=\hat a_{\mathrm{RF}}^2$ satisfies $|\hat a-a|\leq2\sqrt{a}\,\epsilon+\epsilon^2$. It is not, however, stronger than the guarantee of \cref{theorem:main} anywhere. Since $\sin$ is $1$-Lipschitz, $|\hat\theta-\lambda|\leq\epsilon$ immediately gives $|\hat a_{\mathrm{RF}}-a_{\mathrm{RF}}|\leq\epsilon$, so \cref{theorem:main} certifies their metric at no extra cost, while the converse fails as $a\to1$, because inverting $\lambda=\arcsin(a_{\mathrm{RF}})$ costs a factor $(1-a)^{-1/2}$, so an $\epsilon$-additive estimate of $a_{\mathrm{RF}}$ certifies only $|\hat\theta-\lambda|\lesssim\epsilon(1-a)^{-1/2}$. 
Moreover, their overheads are explicitly amplitude-dependent. On a device limited to coherent depth $M$, the depth requirement $M\gtrsim a_{\mathrm{RF}}^{-1/(1-\beta)}$ restricts the algorithm to amplitudes $a_{\mathrm{RF}}\gtrsim M^{-(1-\beta)}$, a floor that rises as the circuits become shallower, i.e. as $\beta$ increases, and that excludes the classical endpoint at $\beta=1$. \cref{theorem:main}, by contrast, holds for every $a\in[0,1]$ and every $\beta\in[0,1]$ with no amplitude-dependent depth or query overhead. 

\bigskip
Vu, Cheng, and Rebentrost propose a general construction of low-depth amplitude estimation algorithms by applying classical Monte Carlo methods to a weakly biased quantum amplitude estimator \cite{vu2025low}. For every $\beta\in [0,1]$, their algorithms find, with probability at least $1-\delta$, an estimate $\hat{a}$ of $a$ satisfying $|\hat{a}-a|\leq \epsilon$, using $\tilde{\mathcal{O}}(\epsilon^{-1-\beta}\log\delta^{-1})$ queries at maximal depth $\tilde{\mathcal{O}}(\epsilon^{-1+\beta})$. Their construction therefore provides a continuous nearly-optimal interpolation for additive amplitude error, but does not establish the nearly-optimal tradeoff for uniform accuracy in the angle $\lambda$. Therefore, their stated guarantees cannot certify a quantum speedup at the boundary $a\to 0$ or $1$. Moreover, the concrete weakly biased estimators available for their construction require a QFT, whereas WLSAE obtains the full interpolation directly for the angle $\lambda$ using neither ancillas nor controlled Grover operations.

\bigskip
Huang and Koczor developed GLSAE, the Gaussian member of WLSAE, and provide angle guarantees $|\hat{\theta}-\lambda|\leq \epsilon$ only away from the boundary \cite{huang2026low}. At the boundary their analysis instead falls back to classical sampling, thereby only providing uniform additive guarantees on $a$. They argue that this patch is necessary because of the flattening loss near the boundary, which is caused by only sampling cosine signals, i.e. $Z_m$ with $\mathbb{E}[Z_m]=\cos(2m\lambda)$, instead of also sampling signals $X_m$ with $\mathbb{E}[X_m]=\sin(2m\lambda)$. The latter would remove the flattening, since $\partial_\lambda \sin(2m\lambda)$ does not vanish as $\lambda\to 0$ or $\pi/2$. However, obtaining these signals, for example by readout with the Hadamard test, requires an ancilla qubit and controlled Grover operations. While it is true that the loss flattens at the boundary, we show that the noise vanishes at the same rate, leaving us with a $\lambda$-independent signal-to-noise ratio (see \cref{eq:snr}), and therefore a $\lambda$-independent localization condition yielding the angle-uniform guarantee of \cref{theorem:main}. Notably, by showing that the cosine signals are sufficient on their own, this work removes the apparent caveat of their ancilla-free measurement model. 
Furthermore, we argue that, while the Gaussian window is an admissible window and a popular choice for applications in spectral estimation, it might not be ideal for amplitude estimation. The reason is that tapered windows, such as the Gaussian or the Kaiser window, suppress side lobes at the cost of a wider main lobe in order to suppress leakage from neighboring frequencies. The AE signal, however, contains a single frequency, so side lobe suppression is not required and the widened main lobe only weakens the localization condition. We therefore argue that tapered windows are counterproductive for AE and suggest using uniform or anti-tapered windows instead, a claim supported by \cref{fig:1}.

\bigskip
Inspired by GLSAE, Sun \textit{et al.} introduced Gaussian-sampled Chebyshev amplitude estimation (GCAE) \cite{sun2026quantum}, with the algorithm and its analysis deferred to a companion work still in preparation at the time of writing. As in GLSAE, they sample depths $m$ from a truncated Gaussian and, depending on the parity of $m$, measure $2\ket{\psi}\bra{\psi}-I$ on $(I-2P)G^{t-1}\ket{\psi}$ or $I-2P$ on $G^t\ket{\psi}$, obtaining a signal $Z_m\in\{-1,1\}$ satisfying \cref{eq:signal} from the same circuits as WLSAE. Instead of parameterizing in angle space, however, they parametrize in amplitude space, fitting $b:=1-2a$ through the loss
\begin{equation}
    L_\mathrm{GCAE}(x)=\frac{1}{N}\sum^N_{i=1}\left(Z_i-T_{m_i}(x) \right)^2,\qquad x\in[-1,1],
\end{equation}
where $T_m(x):=\cos(m\arccos(x))$ denotes the Chebyshev polynomial of the first kind. Substituting $x=\cos(2\theta)$ gives $T_m(x)=\cos(2m\theta)$ and recovers \cref{eq:loss} exactly, so their loss is a simple reparameterization of ours. Near the boundary, the map $\lambda \mapsto a=\sin^2\lambda$ compresses the parameter scale quadratically, so a loss that is quartic in $\theta$ becomes quadratic in $x$, and they report a guarantee uniform over $b\in[-1,1]$, equivalently over all $a\in[0,1]$, without a separate boundary patch.
This gain however does not survive the change of metric. Close to the boundary, $\lambda\to0$ or $\pi/2$, an amplitude guarantee $|\hat{a}-a|\leq \epsilon$ certifies only $|\hat{\theta}-\lambda|\leq \sqrt{\epsilon}$. The quadratic gain in the loss is thus exactly undone by the square root in the conversion---both effects originating from the same reparameterization. Also in the amplitude metric the angle-parametrized guarantee is the stronger one. At equal query cost \cref{theorem:main} certifies $|\hat{a}-a|\leq 2\sqrt{a(1-a)}\epsilon+\epsilon^2$, which is smaller than the flat $\epsilon$ by a factor $2\sqrt{a(1-a)}+\epsilon$ away from $a=1/2$, with the first term vanishing as $a\to0$ or $1$. Using variance-aware concentration inequalities, we showed that the quadratic loss bought by the reparameterization is not needed: WLSAE requires no boundary patch either, and attains the stronger guarantee, uniform across all $\lambda\in[0,\pi/2]$.

\section{Proof of Main Results}
\label{app:proof}
\begin{definition}[Window]
    Let $M\in\mathbb{N}_{\geq 1}$ and $W_M:=\{-M,...,M\}$. A function $p:W_M\to[0,1]$ is called window if $\sum_{m\in W_M}p(m)=1$.
\end{definition}
\begin{definition}[Admissible window]
    \label{def:admissible_app}
    Let $M\in \mathbb{N}_{\geq 1}$, let $\sigma\geq 1$, $b\in(0,1]$, and $C,C',C'',C'''>0$ be constants independent of $M$, and set $T:=\max(1,M/\sigma)$. A window $p$ on $W_M$ is called $(b,\sigma, C, C',C'',C''')$-admissible if 
    \begin{enumerate}
        \item[(W1)] $p(m)+p(-m)\geq \frac{b}{T}$ for all $m\in W_M$ with $1\leq m \leq T$, and
        \item[(W2)] $\mathbb{E}[m^2]\leq CT^2$, $\mathbb{E}[m^4]\leq C'T^4$, $\mathbb{E}[m^6]\leq C''T^6$, and $\mathbb{E}[m^8]\leq C'''T^8$
    \end{enumerate}
\end{definition}
\begin{remark}
\label{remark:definitions}
\Cref{def:admissible_app} refines \cref{def:admissible} of the main text: condition (W1) is identical, while (W2) makes the moment bounds explicit so that we can track constant factors. Condition (W2) imposes no additional restriction. In particular, every window on $W_M$ is supported on $|m| \leq M \leq  \sigma T$, so $\mathbb{E}[m^{2k}] \leq \sigma^{2k} T^{2k}$ for every $k$, and hence every $(b,\sigma)$-admissible window is $(b,\sigma,\sigma^2,\sigma^4,\sigma^6,\sigma^8)$-admissible. Conversely, every $(b,\sigma,C,C',C'',C''')$-admissible window is $(b,\sigma)$-admissible by (W1). The two definitions therefore describe the same class of windows, and every statement below applies verbatim to a $(b,\sigma)$-admissible window with $C = \sigma^2$, $C' = \sigma^4$, $C'' = \sigma^6$, $C''' = \sigma^8$.
\end{remark}

\begin{definition}[optimal loss]
    Let $p$ be a window and $\lambda\in [0,\pi/2]$. The optimal loss $\mathcal{E}^{(\lambda)}:[0,\pi/2]\to [0,4]$ is defined via
    \begin{equation}
        \mathcal{E}^{(\lambda)}(\theta)=\mathbb{E}_{m\sim p}\left[\mathcal{E}^{(\lambda)}_m(\theta) \right],\qquad\mathcal{E}_m^{(\lambda)}(\theta)= (\cos(2m\lambda)-\cos(2m\theta))^2.
    \end{equation}
\end{definition}

\begin{definition}[empirical loss]
    \label{def:empirical_loss}
    Let $p$ be a window, $\lambda\in [0,\pi/2]$, and $N\in\mathbb{N}_{\geq 1}$. For $m\in W_M$, define $p^{(\lambda)}_{m}:\{-1,1\}\to[0,1]$ via $p^{(\lambda)}_m(-1)=\sin^2(m\lambda)$ and $p^{(\lambda)}_m(1)=\cos^2(m\lambda)$. Given $N$ i.i.d.\ draws $\left(m_i,Z_{i} \right)$ with $m_i\sim p$ and, conditional on $m_i$, $Z_{i}\sim p^{(\lambda)}_{m_i}$, define the empirical loss $L^{(\lambda)}_N:[0,\pi/2]\to [0,4]$ via
    \begin{equation}
        L^{(\lambda)}_N(\theta)=\frac{1}{N}\sum_{i=1}^N\left(Z_{i}-\cos(2m_i\theta)\right)^2.
    \end{equation}
    Furthermore, define the empirical fluctuation as $D^{(\lambda)}_N(\theta):=L^{(\lambda)}_N(\theta)-L^{(\lambda)}_N(\lambda)-\mathcal{E}^{(\lambda)}(\theta)$.
\end{definition}

\begin{definition}[uniform grid]
    \label{def:uniform_grid}
    For $K\in\mathbb{N}_{\geq1}$, let $\kappa:=\frac{\pi}{2K}$ and define the uniform grid of spacing
    $\kappa$ with $K$ grid points as
    $G_K:=\{(j+\frac{1}{2})\kappa:0\leq j< K\}\subset[0,\pi/2]$. 
\end{definition}

\begin{remark}
    A few remarks are in order:
    \begin{enumerate}
        \item[(R1)] It is $\mathbb{E}_{Z_m\sim p_m^{(\lambda)}}[Z_m]=\cos(2m\lambda)$ and $\mathrm{Var}_{Z_m\sim p_m^{(\lambda)}}[Z_m]=\sin^2(2m\lambda)$.
        \item[(R2)] The covering radius of a uniform grid $G_K$ with $\kappa=\frac{\pi}{2K}$ is $\kappa/2$: for every $\lambda\in[0,\pi/2]$ there exists a $\theta\in G_K$  with $|\theta-\lambda|\leq\kappa/2$.
    \end{enumerate}
\end{remark}

\begin{proposition}
    Let $p$ be a window, $\lambda\in [0,\pi/2]$, and $N\in\mathbb{N}_{\geq 1}$. Then, for every $\theta\in[0,\pi/2]$, we have
    \begin{equation}
        \mathbb{E}\left[L^{(\lambda)}_N(\theta)-L^{(\lambda)}_N(\lambda)\right]=\mathcal{E}^{(\lambda)}(\theta),
    \end{equation}
    which implies $\mathbb{E}[D^{(\lambda)}_N(\theta)]=0$.
\end{proposition}
\begin{proof}
    It is 
    \begin{equation}
        (Z_{i}-\cos(2m_i\phi))^2=1-2Z_{i}\cos(2m_i\phi)+\cos^2(2m_i\phi),
    \end{equation}
    and, from (R1), we know that $\mathbb{E}_{Z_{i}\sim p_{m_i}^{(\lambda)}}\left[Z_{i}\right]=\cos(2m_i\lambda)$. It follows
    \begin{align}
        \label{eq:prop_1}
        \mathbb{E}_{Z_{i}\sim p_{m_i}^{(\lambda)}}\left[ (Z_{i}-\cos(2m_i\phi))^2 \right] = 1-2\cos(2m_i\lambda) \cos(2m_i\phi)+\cos^2(2m_i\phi).
    \end{align}
    Using that $\mathbb{E}\left[\,\cdot \, \right]:=\mathbb{E}_{m\sim p}\left[\mathbb{E}_{Z_{m}\sim p_{m}^{(\lambda)}}(\,\cdot \,) \right]$, it follows
    \begin{align}
        \mathbb{E}\left[L^{(\lambda)}_N(\theta)-L^{(\lambda)}_N(\lambda)\right]&=\frac{1}{N}\sum^N_{i=1}\mathbb{E}_{m_i\sim p}\left[-2\cos(2m_i\lambda) \cos(2m_i\theta)+\cos^2(2m_i\theta)+\cos^2(2m_i\lambda) \right]\\
        &=\frac{1}{N}\sum^N_{i=1}\mathbb{E}_{m_i\sim p}\left[\mathcal{E}^{(\lambda)}_{m_i}(\theta) \right] =\mathcal{E}^{(\lambda)}(\theta),
    \end{align}
    where we have used \eqref{eq:prop_1} twice---once for $\phi=\theta$ and once for $\phi=\lambda$. Using that $\mathbb{E}\left[ \mathcal{E}^{(\lambda)}(\theta)\right]=\mathcal{E}^{(\lambda)}(\theta)$, it follows that $\mathbb{E}[D^{(\lambda)}_N(\theta)]=0$.
\end{proof}

\begin{lemma}[localization]
    \label{lemma:localization}
    Let $p$ be a window, $\lambda\in [0,\pi/2]$, $N\in\mathbb{N}_{\geq 1}$, $K\in\mathbb{N}_{\geq1}$, and $G_K$ the uniform grid of \cref{def:uniform_grid}, with spacing $\kappa$. Let $\hat{\theta}\in\arg\min_{\theta\in G_K}L^{(\lambda)}_N(\theta)$ and $A\geq 1$. If
    \begin{equation}
        \label{eq:localization}
        \max_{\theta\in G_K:|\theta-\lambda|\leq \kappa/2}\left(\mathcal{E}^{(\lambda)}(\theta)+D^{(\lambda)}_N(\theta) \right) < \min_{\theta\in G_K:|\theta-\lambda|\geq A\kappa}\left(\mathcal{E}^{(\lambda)}(\theta)+D^{(\lambda)}_N(\theta) \right),
    \end{equation}
    then $|\hat{\theta}-\lambda|< A\kappa$.
\end{lemma}
\begin{proof}
    From \cref{def:empirical_loss} we know that $L_N^{(\lambda)}(\theta)=D_N^{(\lambda)}(\theta)+L_N^{(\lambda)}(\lambda)+\mathcal{E}^{(\lambda)}(\theta)$. Hence \eqref{eq:localization} implies 
    \begin{equation}
        \max_{\theta\in G_K:|\theta-\lambda|\leq \kappa/2}L_N^{(\lambda)}(\theta) < \min_{\theta\in G_K:|\theta-\lambda|\geq A\kappa}L_N^{(\lambda)}(\theta).
    \end{equation}
    Now, let $\theta_0\in G_K$ with $|\theta_0-\lambda|\leq \kappa/2$, which exists by (R2). For any $\theta\in G_K$ with $|\theta-\lambda|\geq A\kappa$, we have
    \begin{equation}
        L_N^{(\lambda)}(\theta) \geq \min_{\theta'\in G_K:|\theta'-\lambda|\geq A\kappa}L_N^{(\lambda)}(\theta') > \max_{\theta'\in G_K:|\theta'-\lambda|\leq \kappa/2}L_N^{(\lambda)}(\theta')\ge L_N^{(\lambda)}(\theta_0),
    \end{equation}
    so no such $\theta$ attains the minimum. Hence, every $\hat{\theta}\in \arg\min_{\theta\in G_K}L^{(\lambda)}_N(\theta)$ satisfies $|\hat{\theta}-\lambda|<A\kappa$. Notably, we used the convention that $\min \emptyset =+\infty$, so that the lemma holds for all $A\geq 1$ and all $K\in \mathbb{N}\geq 1$.
\end{proof}

In the following we will derive properties of the optimal loss functions of $(b,\sigma, C,C', C'', C''')$-admissible windows. In particular, we are interested in finding lower and upper bounds on $\mathcal{E}^{(\lambda)}(\theta)$ for different parameter settings $\theta, \lambda\in [0,\pi/2]$ with respect to $T$. In the following, we will frequently use (W1), and $\mathcal{E}_0^{(\lambda)}(\theta)=0$ to write
\begin{equation}
    \mathcal{E}^{(\lambda)}(\theta)\geq \frac{b}{T}\sum^{\lfloor T \rfloor}_{m=1}\mathcal{E}_m^{(\lambda)}(\theta),
\end{equation}
as well as the identities 
\begin{equation}
    \label{eq:decomposition_of_loss}
    \mathcal{E}_m^{(\lambda)}(\theta)=4\sin^2(m(\theta+\lambda))\sin^2(m(\theta-\lambda)),
\end{equation}
and, for $x\in[-\frac{\pi}{2}, \frac{\pi}{2}]$
\begin{equation}
    \label{eq:sin_lower_bound}
    \sin^2(x)\geq \frac{4}{\pi^2}x^2.
\end{equation}

\begin{lemma}
    \label{lemma:upper_bound}
    For every $(b,\sigma, C, C', C'', C''')$-admissible window $p$, we have 
    \begin{equation}
        \mathcal{E}^{(\lambda)}(\theta)\leq C_1\min\left(1, T^2|\theta-\lambda|^2, T^4|\theta-\lambda|^2\max(\bar{\lambda}, |\theta-\lambda|)^2 \right),
    \end{equation}
    where $C_1:=4\max(1, C,9C')$ and $\bar{\lambda}:=\min(\lambda, \frac{\pi}{2}-\lambda)$.
\end{lemma}
\begin{proof}
    Define $s:=\min(\theta+\lambda, \pi-\theta-\lambda)\in [0,\pi/2]$. It is 
\begin{equation}
    \label{eq:reparametrization}
    \sin^2(ms)=\sin^2(m\pi -m(\theta+\lambda))=\sin^2(m(\theta+\lambda)),
\end{equation}
where we have used that $\sin^2$ is $\pi$-periodic and symmetric around $0$. Furthermore, for any $x,y\geq 0$ we have 
\begin{equation}
    \max(y, |x-y|)\leq x+y\leq 3\max(y, |x-y|).
\end{equation}
While the first inequality is trivial, the second follows from $x+y=(x-y)+2y\leq |x-y|+2y\leq 3\max(y, |x-y|)$. It follows
\begin{equation}
    \label{eq:min_eq_1}
    \max(\lambda, |\theta-\lambda|)\leq \theta+\lambda\leq 3\max(\lambda, |\theta-\lambda|),
\end{equation}
and
\begin{equation}
    \label{eq:min_eq_2}
    \max(\frac{\pi}{2}-\lambda, |\theta-\lambda|)=\max(\frac{\pi}{2}-\lambda, |\frac{\pi}{2}-\theta+\lambda-\frac{\pi}{2}|)\leq \pi-\theta-\lambda\leq 3\max(\frac{\pi}{2}-\lambda, |\frac{\pi}{2}-\theta+\lambda-\frac{\pi}{2}|)=3\max(\frac{\pi}{2}-\lambda, |\theta-\lambda|). 
\end{equation}
Taking the minimum of \eqref{eq:min_eq_1} and \eqref{eq:min_eq_2} and using $\min(\max(a,c), \max(b,c))=\max(\min(a,b),c)$ , it follows
\begin{equation}
    \label{eq:barlambdaidenity}
    \max(\bar{\lambda}, |\theta-\lambda|)\leq s \leq 3\max(\bar{\lambda}, |\theta-\lambda|),
\end{equation}
where $\bar{\lambda}:=\min(\lambda, \frac{\pi}{2}-\lambda)$. Moreover, we have $\sin^2(x)\leq \min(1,x^2)$ for all $x\in\mathbb{R}$. 
    Using \eqref{eq:decomposition_of_loss}, \eqref{eq:reparametrization}, and \eqref{eq:barlambdaidenity}, it follows
\begin{align}    
    \mathcal{E}_m^{(\lambda)}(\theta)&\leq 4\min\left(1, m^2s^2\right)\min\left(1, m^2(\theta-\lambda)^2\right)\\&\leq 4\min\left(1, m^2|\theta-\lambda|^2, 9m^2\max(\bar\lambda, |\theta-\lambda|)^2,  9m^4|\theta-\lambda|^2\max(\bar\lambda, |\theta-\lambda|)^2 \right)\\
    \label{eq:upper_bound_per_term_loss}
    &\leq 4\min\left(1, m^2|\theta-\lambda|^2,  9m^4|\theta-\lambda|^2\max(\bar\lambda, |\theta-\lambda|)^2 \right),
\end{align}
where we used that we always have $|\theta-\lambda|^2\leq 9\max(\bar \lambda, |\theta-\lambda|)^2$. Using (W2) and $\mathbb{E}[\min(\,\cdot\,,\,\cdot\,)]\leq \min(\mathbb{E}[\,\cdot\,], \mathbb{E}[\,\cdot\,])$, it follows
\begin{align}
    \mathcal{E}^{(\lambda)}(\theta)&\leq4\min\left(1,CT^2|\theta-\lambda|^2, 9C'T^4|\theta-\lambda|^2\max(\bar \lambda, |\theta-\lambda|)^2 \right)\\&\leq C_1\min\left(1, T^2|\theta-\lambda|^2, T^4|\theta-\lambda|^2\max(\bar \lambda, |\theta-\lambda|)^2 \right),
\end{align}
where $C_1:=4\max\left(1, C, 9C' \right)$.
\end{proof}

\begin{lemma}
    \label{lemma:lower_bound_near}
    For every $(b,\sigma, C, C', C'', C''')$-admissible window $p$ and $|\theta-\lambda|\leq \frac{\pi}{2T}$, we have 
    \begin{equation}
         \mathcal{E}^{(\lambda)}(\theta)\geq C_2T^2|\theta-\lambda|^2\min\left(1, T^2\max(\bar{\lambda}, |\theta-\lambda|)^2\right),
    \end{equation}
    where $C_2:=\frac{2b}{1215\pi^4}$ and $\bar{\lambda}:=\min(\lambda, \frac{\pi}{2}-\lambda)$.
\end{lemma}
\begin{proof}
Define $s:=\min(\theta+\lambda, \pi-\theta-\lambda)\in [0,\pi/2]$ and $\bar{\lambda}:=\min(\lambda, \frac{\pi}{2}-\lambda)$. For $|\theta-\lambda|\leq \frac{\pi}{2T}$ we can then use \eqref{eq:sin_lower_bound} and \eqref{eq:reparametrization} to obtain
    \begin{equation}
        \mathcal{E}^{(\lambda)}(\theta)\geq \frac{16b}{\pi^2T}|\theta-\lambda|^2\sum^{\lfloor T \rfloor}_{m=1}m^2\sin^2(ms).
    \end{equation}
    To proceed we will temporarily distinguish between the following two cases.
    \begin{enumerate}
        \item $T\in [1,3]$. We then have
        \begin{equation}
            \mathcal{E}^{(\lambda)}(\theta)\geq \frac{16b}{27\pi^2}T^2|\theta-\lambda|^2\sin^2(s).
        \end{equation}
        \begin{enumerate}
            \item[(i)] If $s\leq \frac{\pi}{2T}$, we use \eqref{eq:sin_lower_bound} and \eqref{eq:barlambdaidenity} to get
            \begin{equation}
                \label{eq:first_case}
                \mathcal{E}^{(\lambda)}(\theta)\geq \frac{64b}{243\pi^4}T^4|\theta-\lambda|^2s^2\geq \frac{64b}{243\pi^4}T^4|\theta-\lambda|^2\max(\bar{\lambda}, |\theta-\lambda|)^2.
            \end{equation}
        
            \item[(ii)] If $\frac{\pi}{2T}\leq s\leq \frac{\pi}{2}$, it is
            \begin{equation}
                \mathcal{E}^{(\lambda)}(\theta)\geq \frac{16b}{27\pi^2}T^2|\theta-\lambda|^2\sin^2(\frac{\pi}{2 T})\geq\frac{4b}{27\pi^2}T^2|\theta-\lambda|^2,
            \end{equation}
            where in the last step we have used $\sin^2(\frac{\pi}{2T})\geq \sin^2(\frac{\pi}{2\cdot 3})\geq \frac{1}{4}$.
        \end{enumerate}
        In both regimes the bound reads $\mathcal{E}^{(\lambda)}(\theta)\geq\frac{4b}{27\pi^2}T^2|\theta-\lambda|^2 c$, with $c=1$ in (ii) and $c=\frac{16}{9\pi^2}T^2\max(\bar{\lambda}, |\theta-\lambda|)^2$ in (i). Since $\min\left(1, \frac{16}{9\pi^2}T^2\max(\bar{\lambda}, |\theta-\lambda|)^2\right)$ is at most either value, we obtain
        \begin{equation}
            \label{eq:Tleq3_bound}
            \mathcal{E}^{(\lambda)}(\theta)\geq \frac{4b}{27\pi^2}T^2|\theta-\lambda|^2\min\left(1, \frac{16}{9\pi^2}T^2\max(\bar{\lambda}, |\theta-\lambda|)^2 \right)
        \end{equation}
        for all $\theta,\lambda\in[0,\frac{\pi}{2}]$ with $|\theta-\lambda|\leq \frac{\pi}{2T}$ whenever $T\in[1,3]$.

        \item $T\geq 3$. 
        \begin{enumerate}
            \item[(i)] If $s\leq \frac{3\pi}{2T}$, we use \eqref{eq:sin_lower_bound} to obtain
            \begin{align}
                \mathcal{E}^{(\lambda)}(\theta)&\geq \frac{16b}{\pi^2 T}|\theta-\lambda|^2\sum^{\lfloor T/3\rfloor}_{m=1}m^2\sin^2(ms)\\
                &\geq \frac{64b}{\pi^4 T}|\theta-\lambda|^2\max(\bar{\lambda}, |\theta-\lambda|)^2\sum^{\lfloor T/3\rfloor}_{m=1}m^4\\
                &\geq \frac{64b}{5\pi^4T}|\theta-\lambda|^2\max(\bar{\lambda}, |\theta-\lambda|)^2\lfloor T/3 \rfloor^5\\  &\geq \frac{64b}{5\cdot 6^5\pi^4}T^4|\theta-\lambda|^2\max(\bar\lambda, |\theta-\lambda|)^2\\
                \label{eq:first_bound_forTgeq3}
                &=\frac{2b}{1215\pi^4}T^4|\theta-\lambda|^2\max(\bar\lambda, |\theta-\lambda|)^2,
            \end{align}
            where we have used that
            \begin{equation}
                \sum_{m=1}^Mm^k\geq\int^{M}_0x^k\mathrm{d}x=\frac{M^{k+1}}{k+1}
            \end{equation}
            for every $M\in\mathbb{N}_{\geq 1}$, and $\lfloor T/3\rfloor^5\geq \frac{T^5}{6^5}$.
            
            \item[(ii)] If $\frac{3\pi}{2T}\leq s\leq \frac{\pi}{2}$, we have
            \begin{align}
                \mathcal{E}^{(\lambda)}(\theta)&\geq \frac{16b}{\pi^2T}|\theta-\lambda|^2\sum^{\lfloor T\rfloor}_{m=\lfloor T/3 \rfloor +1}m^2\sin^2(ms) \\
                &\geq \frac{8b}{9\pi^2}T|\theta-\lambda|^2\sum^{\lfloor T\rfloor}_{m=\lfloor T/3 \rfloor +1}(1-\cos(2ms)) \\
                \label{eq:last_line_reference}
                &\geq\frac{8b}{9\pi^2}T|\theta-\lambda|^2\left(\lfloor T \rfloor - \left\lfloor \frac{T}{3}\right\rfloor-\left|\sum^{\lfloor T\rfloor}_{m=\lfloor T/3 \rfloor +1}e^{i2ms}\right|\right),
            \end{align}
            where in the last step we have used that $\mathrm{Re}(z)\leq |\mathrm{Re}(z)|\leq |z|$. It is
            \begin{align}
                \sum^{\lfloor T\rfloor}_{m=\lfloor T/3 \rfloor +1}e^{i2ms}&=e^{i2(\lfloor T/3\rfloor+1)s}\sum_{m=0}^{\lfloor T \rfloor -\lfloor T/3 \rfloor-1}e^{i2ms}\\
                &=e^{i2(\lfloor T/3\rfloor+1)s}\frac{1-e^{i2(\lfloor T\rfloor -\lfloor T/3\rfloor)s}}{1-e^{i2s}},
            \end{align}
            where in the last step we used that the sum is a geometric series. To upper bound its magnitude, we can ignore the phase and notice that the numerator cannot exceed 2. For the denominator, we find
            \begin{equation}
                |1-e^{2is}|=2\sin s\geq 2\sin(\frac{3\pi}{2T})\geq \frac{6}{T}.
            \end{equation}
            It follows
            \begin{equation}
                \left|\sum^{\lfloor T\rfloor}_{m=\lfloor T/3 \rfloor +1}e^{i2ms}\right|\leq T/3.
            \end{equation}
            Inserting this bound in \eqref{eq:last_line_reference}, we obtain
            \begin{equation}
                \mathcal{E}^{(\lambda)}(\theta)\geq \frac{8b}{9\pi^2}T|\theta-\lambda|^2\left(\lfloor T \rfloor - \left\lfloor \frac{T}{3}\right\rfloor-\frac{T}{3}\right)\geq \frac{8b}{9\pi^2}T^2|\theta-\lambda|^2\left(\frac{3}{4}-\frac{2}{3} \right)\geq \frac{2b}{27\pi^2}T^2|\theta-\lambda|^2.
            \end{equation}
        \end{enumerate}
    We merge both regimes to obtain
    \begin{equation}
        \label{eq:Tgeq3_bound}
        \mathcal{E}^{(\lambda)}(\theta)\geq\frac{2b}{27\pi^2}T^2|\theta-\lambda|^2\min\left(1,\frac{1}{45\pi^2}T^2\max(\bar{\lambda}, |\theta-\lambda|)^2 \right)
    \end{equation}
    for all $\theta,\lambda\in[0,\frac{\pi}{2}]$ with $|\theta-\lambda|\leq \frac{\pi}{2T}$ whenever $T\geq 3$.
    \end{enumerate}
    Notably, \eqref{eq:Tgeq3_bound} is weaker than (and therefore implied by) \eqref{eq:Tleq3_bound}. Hence, \eqref{eq:Tgeq3_bound} holds for all $T\geq 1$. Furthermore, since $1\geq \frac{1}{45\pi^2}$, we have
    \begin{equation}
        \mathcal{E}^{(\lambda)}(\theta)\geq C_2T^2|\theta-\lambda|^2\min\left(1, T^2\max(\bar{\lambda}, |\theta-\lambda|)^2 \right),
    \end{equation}
    where $C_2:=\frac{2b}{1215\pi^4}$.
\end{proof}

\begin{lemma}
    \label{lemma:lower_bound_far}
    For every $(b,\sigma, C, C', C'', C''')$-admissible window $p$ and $|\theta-\lambda|\geq \frac{\pi}{2T}$, we have 
    \begin{equation}
        \mathcal{E}^{(\lambda)}(\theta)\geq C_3,
    \end{equation}
    where $C_3:=\frac{4}{243}b$.
\end{lemma}
\begin{proof}
    To proceed we will temporarily distinguish between the following two cases.
    \begin{enumerate}
        \item Let $T\geq 3$.
    $|\theta-\lambda|\geq \frac{\pi}{2T}$ implies $\theta+\lambda \in [\frac{\pi}{2T}, \pi-\frac{\pi}{2T}]$. We can expand
    \begin{equation}
        \mathcal{E}^{(\lambda)}_m(\theta)=1+\frac{1}{2}\left(\cos(4m\theta)+\cos(4m\lambda)\right)-\cos(2m(\theta+\lambda))-\cos(2m(\theta-\lambda))
    \end{equation}
    and use
    \begin{equation}
        \sum_{m=1}^M\cos(mx)=\frac{\sin(x(M+\frac{1}{2}))}{2\sin(x/2)}-\frac{1}{2}
    \end{equation}
    to obtain
    \begin{equation}
        \mathcal{E}^{(\lambda)}(\theta)\geq \frac{b}{2T}\left(2\lfloor T \rfloor+1-D_{\lfloor T \rfloor}(2(\theta+\lambda))-D_{\lfloor T \rfloor}(2(\theta-\lambda))+\frac{1}{2}D_{\lfloor T \rfloor}(4\theta)+\frac{1}{2}D_{\lfloor T \rfloor}(4\lambda) \right),
    \end{equation}
    where 
    \begin{equation}
        D_M(x):=\frac{\sin(x(M+\frac{1}{2}))}{\sin(x/2)}
    \end{equation}
    is the Dirichlet kernel. We need a lower bound on $D_M(x)$ for $x\in [0,2\pi]$ and an upper bound on $D_M(x)$ for $x\in[\frac{\pi}{T}, 2\pi-\frac{\pi}{T}]$. We start by finding the lower bound. Let $L:=M+\frac{1}{2}$ and consider the following cases:
    \begin{enumerate}
        \item $x\in [0,\frac{\pi}{L}]$. Here we have $D_M(x)\geq 0$
        \item $x\in [\frac{\pi}{L}, \pi]$. Here we have $0\leq x/2\leq \pi/2$. It follows
        \begin{equation}
            \frac{1}{\sin(x/2)}\leq \frac{\pi}{x}. 
        \end{equation}
        Hence, we have
        \begin{equation}
            D_M(x)\geq \min\left(0, \pi \frac{\sin(xL)}{x} \right).
        \end{equation}
        With 
        \begin{equation}
            \pi \frac{\sin(xL)}{x}=L\pi\frac{\sin(xL)}{xL}\geq L\pi \inf_{\pi\leq y \leq L\pi}\frac{\sin y}{y}\geq L\pi \inf\frac{\sin y}{y}> -0.7L, 
        \end{equation}
        it follows $D_M(x)\geq -0.7L$.
        \item $x\in [\pi, 2\pi]$. $\sin(x)$ is symmetric around $x=\pi/2+k\pi$ for $k\in\mathbb{Z}$. Hence, both $\sin(x/2)$ and $\sin(Lx)=\sin((M+1/2)x)$ are symmetric around $\pi$. Therefore, also $D_M(x)$ is symmetric around $\pi$ and we obtain the same bounds as in the other two cases.
    \end{enumerate}
    Consequently, we have $D_M(x)\geq -0.7L$ for all $x\in[0,2\pi]$ (and also all $x\in\mathbb{R}$). We now derive the upper bound. Because of the symmetry of $D_M(x)$ around $\pi$, it is sufficient to find an upper bound for $x\in [\frac{\pi}{T}, \pi]$. In this region we can write
    \begin{equation}
        D_M(x)\leq \max\left(0, \pi\frac{\sin(xL)}{x} \right).
    \end{equation}
    With
    \begin{equation}
        \pi\frac{\sin(xL)}{x}\leq L\pi\sup_{\frac{\pi L}{T}\leq y}\frac{\sin y}{y} .
    \end{equation}
    For $L=\lfloor T \rfloor +1/2$ and $T\geq 3$, we have $\frac{\pi L}{T}=\pi \frac{\lfloor T\rfloor+\frac{1}{2}}{T}\geq \frac{7}{8}\pi$ and 
    \begin{equation}
        \pi\sup_{\frac{7\pi}{8}\leq y}\frac{\sin y}{y}<0.44.
    \end{equation}
    Hence, $D_M(x)\leq 0.44L=0.44(\frac{1}{2}+\lfloor T \rfloor)$. Putting everything together, for $T\geq 3$, it follows
    \begin{align}
        \mathcal{E}^{(\lambda)}(\theta)&\geq \frac{b}{2T}(1-0.44-0.35)\left(2\lfloor T \rfloor+1\right)\\
        &=0.21b\frac{\lfloor T\rfloor +\frac{1}{2}}{T}> 0.18b.
    \end{align}
    \item Let $T\in [1,3]$. As before, we can use that $|\theta-\lambda|\geq \frac{\pi}{2T}$ implies $\theta+\lambda \in [\frac{\pi}{2T}, \pi-\frac{\pi}{2T}]$ to find
    \begin{align}
        \mathcal{E}^{(\lambda)}(\theta)&\geq \frac{4b}{T}\sum^{\lfloor T \rfloor}_{m=1}\sin^2(m(\theta+\lambda))\sin^2(m(\theta-\lambda))\\
        &\geq \frac{4b}{T}\sin^2(\theta+\lambda)\sin^2(\theta-\lambda)\\
        &\geq \frac{4b}{T^5}\geq \frac{4b}{3^5}=\frac{4}{243}b.
    \end{align}
    \end{enumerate}
    Consequently, for $|\theta-\lambda|\geq \frac{\pi}{2T}$, we have $\mathcal{E}^{(\lambda)}(\theta)\geq \frac{4}{243}b$ for all $T\geq 1$.
\end{proof}

\begin{corollary}
    \label{cor:lower_bound}
    For every $(b,\sigma, C,C', C'', C''')$-admissible window $p$ we have
    \begin{equation}
        \mathcal{E}^{(\lambda)}(\theta)\geq C_2\min\left(1, T^2|\theta-\lambda|^2, T^4|\theta-\lambda|^2\max(\bar \lambda, |\theta-\lambda|)^2 \right),      
    \end{equation}
    where $C_2:=\frac{2b}{1215\pi^4}$ and $\bar\lambda:=\min(\lambda, \frac{\pi}{2}-\lambda)$.
\end{corollary}
\begin{proof}
    We can combine \cref{lemma:lower_bound_near} and \cref{lemma:lower_bound_far} to obtain 
    \begin{align}
        \mathcal{E}^{(\lambda)}(\theta)&\geq \min\left(C_3,C_2T^2|\theta-\lambda|^2\min(1, T^2\max(\bar\lambda, |\theta-\lambda|)^2)  \right)\\
        &\geq C_2\min\left(1, T^2|\theta-\lambda|^2, T^4|\theta-\lambda|^2\max(\bar\lambda, |\theta-\lambda|)^2 \right)
    \end{align}
    for all $\theta, \lambda\in [0,\frac{\pi}{2}]$, where we have used that $C_2<C_3$.
\end{proof}
\begin{corollary}
    \label{cor:envelope}
    For every $(b, \sigma,C, C', C'', C''')$-admissible window $p$ we have
    \begin{equation}
        C_2\Psi(|\theta-\lambda|)^2\leq \mathcal{E}^{(\lambda)}(\theta)\leq C_1\Psi(|\theta-\lambda|)^2,
    \end{equation}
    where $\Psi(\Delta):=\min\left(1, T\Delta, T^2\Delta\max(\bar \lambda, \Delta) \right)$ and $\bar\lambda:=\min(\lambda, \frac{\pi}{2}-\lambda)$.
\end{corollary}
\begin{proof}
    Immediate consequence of \cref{lemma:upper_bound} and \cref{cor:lower_bound}.
\end{proof}

\begin{lemma}[concentration of empirical fluctuation]
    \label{lemma:concentration}
    Let $p$ be a $(b,\sigma,C, C', C'', C''')$-admissible window, $\lambda\in [0,\pi/2]$, $N\in\mathbb{N}_{\geq 1}$, $\delta\in(0,1)$, $K\in\mathbb{N}_{\geq 1}$, and $G_K$ the uniform grid of \cref{def:uniform_grid}. Then, with probability at least $1-\delta$, we have
    \begin{equation}
        |D_N^{(\lambda)}(\theta)|\leq C_{8}\max\left(r\Psi(|\theta-\lambda|), \sqrt{r}\Omega(|\theta-\lambda|) \right)
    \end{equation}
    for all $\theta\in G_K$, where $r=\frac{\log(2K/\delta)}{N}$, $C_8:=\sqrt{32}\max\left(24\sigma^2, 4C, 36C',12\sqrt{2C''}, 18\sqrt{2C'''}\right)$, and 
    \begin{equation}
        \Psi(\Delta):=\min\left(1, T\Delta, T^2\Delta\max(\bar\lambda, \Delta) \right),\qquad\Omega(\Delta):=\min\left(1, \min\left( 1,T^2\max(\bar\lambda, \Delta)^2\right)\max\left[T\min(\bar\lambda, \Delta), T^2\Delta^2\right] \right),
    \end{equation}
    where $\bar\lambda:=\min(\lambda, \frac{\pi}{2}-\lambda)$
\end{lemma}

\begin{proof}
    For $i\in\{1,...,N\}$, define $X_i(\theta):=(Z_{i}-\cos(2m_i\theta))^2-(Z_{i}-\cos(2m_i\lambda))^2-\mathcal{E}^{(\lambda)}(\theta)$, such that $D_N^{(\lambda)}(\theta)=\frac{1}{N}\sum_{i=1}^NX_i(\theta)$. For now we drop the index $i$ and treat $X(\theta)$ as a single random variable. Using $Z_{m}^2=1$, we can write
    \begin{align}
        X(\theta)&=2Z_{m}(\cos(2m\lambda)-\cos(2m\theta))+\cos^2(2m\theta)-\cos^2(2m\lambda)-\mathcal{E}^{(\lambda)}(\theta)\\
        &=2(\cos(2m\lambda)-Z_{m})(\cos(2m\theta)-\cos(2m\lambda))-\mathcal{E}^{(\lambda)}(\theta)+\mathcal{E}_{m}^{(\lambda)}(\theta),
    \end{align}
    to see that $\mathbb{E}[X(\theta)]=0$. Using the law of total variance, it follows
    \begin{align}
        \label{eq:var}
        \mathrm{Var}\left[X(\theta) \right]&=\mathbb{E}_{m\sim p}\left[\mathrm{Var}\left[X(\theta)\mid m \right] \right]+\mathrm{Var}_{m\sim p}\left[\mathbb{E}\left[X(\theta)\mid m \right] \right]\\
        &=\mathbb{E}_{m\sim p}\left[\mathrm{Var}_{Z_{m}\sim p_{m}^{\lambda}}\left[X(\theta) \right] \right]+\mathrm{Var}_{m\sim p}\left[\mathbb{E}_{Z_{m}\sim p_{m}^{\lambda}}\left[X(\theta)\right] \right],
    \end{align}
    where, using (R1), $\mathbb{E}_{Z_{m}\sim p_{m}^{\lambda}}\left[X(\theta)\right]=-\mathcal{E}(\theta)+\mathcal{E}_m(\theta)$ and $\mathrm{Var}_{Z_{m}\sim p_{m}^{\lambda}}\left[X(\theta) \right]=4\mathcal{E}^{(\lambda)}_m(\theta)\mathrm{Var}_{Z_{m}\sim p_{m}^{\lambda}}\left[Z_{m} \right]=4\mathcal{E}^{(\lambda)}_m(\theta)\sin^2(2m\lambda)$. The first term on the RHS of \cref{eq:var} corresponds to the measurement noise of the binary variable $Z_m$, while the second term captures the additional variance induced by randomly sampling $m$. As before (see \eqref{eq:upper_bound_per_term_loss}) we obtain
    \begin{equation}
        \mathcal{E}^{(\lambda)}_m(\theta)\leq 4\min\left(1, m^2|\theta-\lambda|^2,  9m^4|\theta-\lambda|^2v^2 \right),
    \end{equation}
    where we have defined $v:=\max(\bar \lambda, |\theta-\lambda|)$ and $\sin^2(2m\lambda)=\sin^2(2m\bar\lambda)\leq \min(1, 4m^2\bar{\lambda}^2)$. It follows
    \begin{equation}
        \mathrm{Var}_{Z_{m}\sim p_{m}^{\lambda}}\left[X(\theta)\right]\leq 16\min\left(1, m^2|\theta-\lambda|^2, 9m^4|\theta-\lambda|^2v^2, 4m^2\bar\lambda^2, 4m^4|\theta-\lambda|^2\bar\lambda^2, 36m^6|\theta-\lambda|^2v^2\bar\lambda^2 \right).
    \end{equation}
    Notably, we always have $4m^4\bar\lambda^2\leq 9m^4v^2$, which allows us to drop the $9m^4|\theta-\lambda|^2v^2$ term. Using (W2), it follows
    \begin{align}
        \mathbb{E}_{m\sim p}\left[\mathrm{Var}_{Z_{m}\sim p_{m}^{\lambda}}\left[X(\theta) \right] \right]&\leq 16\min\left(1, CT^2\min(|\theta-\lambda|, 2\bar\lambda)^2, 4C'T^4|\theta-\lambda|^2\bar\lambda^2, 36C''T^6|\theta-\lambda|^2v^2\bar\lambda^2\right)\\
        &\leq C_4\min\left(1, T^2\min(|\theta-\lambda|, \bar\lambda)^2, T^4|\theta-\lambda|^2\bar\lambda^2, T^6|\theta-\lambda|^2v^2\bar\lambda^2 \right),
    \end{align}
    where $C_4:=16\max(1, 4C, 4C', 36C'')$. Notably, we can also drop the $T^4|\theta-\lambda|^2\bar\lambda^2$ term. To see this notice that $\min(a,b)\max(a,b)=ab$ for all $a,b\in \mathbb{R}$. Define $u:=\min(|\theta-\lambda|,\bar\lambda)$. Then $|\theta-\lambda|^2\bar\lambda^2=v^2u^2$. Therefore, we have
    \begin{align}
        \min\left(T^2u^2, T^4|\theta-\lambda|^2\bar{\lambda}^2, T^6|\theta-\lambda|^2v^2\bar\lambda^2 \right)=&\min\left( T^2u^2, T^4u^2v^2, T^6u^2v^4\right)\\
        =&T^2u^2\min\left(1, T^2v^2, T^4v^4 \right).
    \end{align}
    If $Tv\geq 1$, the constant term is the smallest. If $Tv<1$, the quartic term is the smallest. Hence, in neither case we need the second term so that we can drop it. Furthermore, (W2) gives us
    \begin{align}
        \mathrm{Var}_{m\sim p}\left[\mathbb{E}_{Z_{m}\sim p_{m}^{\lambda}}\left[X(\theta)\right] \right]&\leq \mathbb{E}_{m\sim p}[(\mathcal{E}^{(\lambda)}_m)^2]\\
        &\leq 16\min\left(1, C'T^4|\theta-\lambda|^4, 81C'''T^8|\theta-\lambda|^4v^4 \right)\\
        &\leq C_5\min\left(1, T^4|\theta-\lambda|^4, T^8|\theta-\lambda|^4v^4 \right),
    \end{align}
    where $C_5:=16\max(1, C', 81C''')$. For notational convenience set $\Delta:=|\theta-\lambda|$ and $u:=\min(\Delta, \bar\lambda)$. As before, we will use that $uv=\Delta\bar\lambda$.
    With $C_6:=2\max(C_4,C_5)=32\max(1,4C, 4C', 36C'', 81C''')$, it follows
    \begin{align}
        \mathrm{Var}\left[X(\theta) \right]&\leq 2\max\left(\mathbb{E}_{m\sim p}\left[\mathrm{Var}_{Z_{m}\sim p_{m}^{\lambda}}\left[X(\theta) \right] \right], \mathrm{Var}_{m\sim p}\left[\mathbb{E}_{Z_{m}\sim p_{m}^{\lambda}}\left[X(\theta)\right] \right] \right)\\
        &\leq C_6\max\left[\min\left(1, T^2u^2, T^6v^4u^2 \right),  \min\left(1, T^4\Delta^4, T^8\Delta^4v^4 \right)\right]\\
        &= C_6\min\left(1,\max \left[T^2u^2\min\left(1,T^4v^4 \right),T^4\Delta^4\min\left(1, T^4v^4 \right) \right]\right)\\
        &=C_6\min\left(1, \min\left( 1,T^4v^4\right)\max\left[T^2u^2, T^4\Delta^4\right] \right)=:V.
    \end{align}
    Moreover, we find
    \begin{align}
        |X(\theta)|&\leq \left|\left(-2Z_m+(\cos(2m\lambda)+\cos(2m\theta) )\right)\left(\cos(2m\theta)-\cos(2m\lambda) \right)\right|+\mathcal{E}^{(\lambda)}(\theta)\\
        &\leq 8\left|\sin(m(\theta-\lambda))\sin(m(\theta+\lambda)) \right|+\mathcal{E}^{(\lambda)}(\theta)\\
        &\leq 8\min\left(1,m|\theta-\lambda|, 3m^2|\theta-\lambda|v \right)+C_1\min\left(1, T^2|\theta-\lambda|^2, T^4|\theta-\lambda|^2v^2\right)\\
        &\leq C_7\min\left(1,T|\theta-\lambda|, T^2|\theta-\lambda|v \right)=:B,
    \end{align}
    where we have used \eqref{eq:reparametrization}, \eqref{eq:barlambdaidenity}, and $C_7:=2\max(C_1, 24\sigma^2)$. For any fixed $\theta\in G_K$, Bernstein's inequality yields, for every $\eta>0$,
    \begin{equation}
        \mathrm{Pr}\left(\left|D^{(\lambda)}_N(\theta)\right|\geq \eta \right)\leq 2\exp\left(-\frac{N\eta^2}{2\left(V+\frac{B\eta}{3}\right)}\right),
    \end{equation}
    where $B$ and $V$ depend implicitly on $\theta$. To obtain a bound holding simultaneously for all $K$ grid points, for each $\theta\in G_K$ we choose a threshold $\eta=\eta(\theta)$ such that the right-hand side at most $\delta/K$. Equivalently, it suffices that
    \begin{align}
         2K\exp\left(-\frac{N\eta^2}{2\left(V+\frac{B\eta}{3}\right)}\right)&\leq \delta 
    \end{align}
    Defining $L:=\log(2K/\delta)$, this is equivalent to 
    \begin{equation}
         \eta^2-\frac{2BL}{3N}\eta-\frac{2VL}{N}\geq 0,
    \end{equation}
    and hence implied by 
    \begin{equation}
        \eta \geq \frac{\frac{2BL}{3N}+\sqrt{(\frac{2BL}{3N})^2+4\frac{2VL}{N}}}{2}=\frac{BL}{3N}+\sqrt{(\frac{BL}{3N})^2+\frac{2VL}{N}}:=\eta^*.
    \end{equation}
    Using $\sqrt{x}+\sqrt{y}\geq \sqrt{x+y}$, it follows that 
    \begin{equation}
        \frac{2BL}{3N}+\sqrt{\frac{2VL}{N}}\geq \eta^*.
    \end{equation}
    Therefore, setting $r=L/N$, each grid point satisfies 
    \begin{equation}
        \mathrm{Pr}\left(\left|D^{(\lambda)}_N(\theta)\right|\geq \frac{2}{3}Br+\sqrt{2Vr}\right)\leq \delta/K.
    \end{equation}
    A union bound over all $\theta\in G_K$ then gives, with probability at least $1-\delta$, simultaneously for every $\theta\in G_K$,
    \begin{align}
        |D^{(\lambda)}_N(\theta)|&\leq \frac{2}{3}Br+\sqrt{2Vr}\leq \sqrt{8}\max\left(Br, \sqrt{Vr} \right)=C_8\max(r\Psi(\Delta), \sqrt{r}\Omega(\Delta)),
    \end{align}
    where 
    \begin{align}
    C_8 &:= \sqrt{8}\max\!\left(\sqrt{C_6}, C_7\right) \\
    &= \sqrt{8}\max\left(4\sqrt2\max(1, 2\sqrt{C}, 2\sqrt{C'}, 6\sqrt{C''}, 9\sqrt{C'''}),
        2\max(4, 4C, 36C', 24\sigma^2)\right) \\
    &= \sqrt{32}\max\left(24\sigma^2, \sqrt{32C}, 4C, \sqrt{32C'}, 36C', 12\sqrt{2C''}, 18\sqrt{2C'''}\right)\\
    &= \sqrt{32}\max\left(24\sigma^2, 4C, 36C',12\sqrt{2C''}, 18\sqrt{2C'''}\right),
    \end{align}
    where we have used $\sigma^2\geq 1$ to absorb the constant entries, together with $\max(24\sigma^2, 4C)\geq \sqrt{32C}$ and $\max(24\sigma^2, 36C')\geq \sqrt{32C'}$. Furthermore, we defined
    \begin{equation}
        \Psi(\Delta):=\min\left(1, T\Delta, T^2\Delta v \right),\qquad\Omega(\Delta):=\min\left(1, \min\left( 1,T^2v^2\right)\max\left[Tu, T^2\Delta^2\right] \right).
    \end{equation}
\end{proof}
\begin{remark}
    \label{remark:variance}
    Define $Y_i(\theta):=X_i(\theta)+\mathcal{E}^{(\lambda)}(\theta)=(Z_i-\cos(2m_i\theta))^2-(Z_i-\cos(2m_i\lambda))^2$ such that $L^{(\lambda)}(\theta)-L^{(\lambda)}(\lambda)=\frac{1}{N}\sum^N_{i=1}Y_i$. Since $Y_i$ and $X_i$ are related by a deterministic additive term, i.e. the optimal loss $\mathcal{E}^{(\lambda)}(\theta)$, both random variables have the same variance and range. For $\Delta\leq 1/T\Leftrightarrow \Delta T\leq 1$ we have 
    \begin{align}
        \Omega(\Delta)&=\min\left( 1,T^2\max(\bar\lambda, \Delta)^2\right)\max\left[T\min(\bar\lambda, \Delta), T^2\Delta^2\right]\leq \min\left(T\Delta,T^3\Delta\max(\bar\lambda, \Delta)^2\right)=\frac{\Psi^2(\Delta)}{T\Delta},
    \end{align}
    implying
    \begin{equation}
        \mathrm{Var}[Y_i]\leq C_6\Omega^2(\Delta)\leq C_6\frac{\Psi^4(\Delta)}{T^2\Delta^2}.
    \end{equation}
    Furthermore, we have
    \begin{equation}
        \max Y_i(\theta)-\min Y_i(\theta)\leq C_7\Psi(\Delta).
    \end{equation}
\end{remark}

\begin{lemma}
    \label{lemma:main}
    Let $M,N\in\mathbb{N}_{\geq 1}$, let $p$ be a $(b,\sigma,C,C',C'',C''')$-admissible window on $\mathcal{W}_M$ with $T=\max(1,M/\sigma)$, and let $\lambda\in[0,\pi/2]$, $\delta\in(0,1)$, and $\epsilon\in(0,1/T]$. Set $A=\frac{3C_8}{8C_2}$ and $C_9=\frac{(1+A)C_8}{C_2A^2-C_1}$, where $C_1$, $C_2$, $C_8$ are defined as in \cref{lemma:upper_bound}, \cref{lemma:lower_bound_near} and \cref{lemma:concentration}, and let $G_K$ be the uniform grid from \cref{def:uniform_grid} with $K=\lceil\frac{A\pi}{2\epsilon}\rceil$. Then, if
    \begin{equation}
    \label{eq:lemma_cond}
    T\sqrt{N}\geq \frac{2AC_9}{\epsilon}\sqrt{\log\left(\frac{2A\pi}{\epsilon\delta}\right)},
    \end{equation}
    any minimizer $\hat{\theta}\in\arg\min_{\theta\in G_K}L^{(\lambda)}_N(\theta)$ satisfies $|\hat{\theta}-\lambda|\leq\epsilon$ with probability at least $1-\delta$.
\end{lemma}

\begin{proof}
Define
\begin{equation}
    U:=\max_{\theta\in G_K:|\theta-\lambda|\leq \kappa/2}\left(\mathcal{E}^{(\lambda)}(\theta)+D^{(\lambda)}_N(\theta) \right)
\end{equation}
and
\begin{align}
    L:&=\min_{\theta\in G_K:|\theta-\lambda|\geq A\kappa}\left(\mathcal{E}^{(\lambda)}(\theta)+D^{(\lambda)}_N(\theta) \right)\\
    &=\min\left[\min_{\theta\in G_K:A\kappa\leq |\theta-\lambda|\leq 1/T}\left(\mathcal{E}^{(\lambda)}(\theta)+D^{(\lambda)}_N(\theta) \right), \min_{\theta\in G_K:|\theta-\lambda|\geq 1/T}\left(\mathcal{E}^{(\lambda)}(\theta)+D^{(\lambda)}_N(\theta) \right)\right]=:\min\left[L_\mathrm{near}, L_\mathrm{far}\right].
\end{align}
To satisfy \eqref{eq:localization} of \cref{lemma:localization}, we require $U< L_\mathrm{near}$ and $U< L_\mathrm{far}$. Let $\Psi$ and $\Omega$ be defined as in \cref{lemma:concentration}, and assume $0<A\kappa\leq \frac{1}{T}\implies 0<\kappa\leq \frac{1}{T}$ and distinguish the following two cases:
\begin{enumerate}
    \item Let $\bar\lambda\leq \frac{1}{T}$. It follows 
    \begin{equation}
        \Psi(\Delta)=T^2\Delta v_\Delta,\qquad\Omega(\Delta)=T^3v_\Delta^2\max(u_\Delta, T\Delta^2)\leq T^3\Delta v^2_\Delta,
    \end{equation}
    for all $\Delta\leq 1/T$, where $v_\Delta:=\max(\bar\lambda, \Delta)$ and $u_\Delta:=\min(\bar\lambda, \Delta)$. Furthermore, assume that $\max(r, \sqrt{r}Tv_\Delta)=\sqrt{r}Tv_\Delta$ for all $\kappa/2\leq \Delta \leq 1/T$. This is implied by $\sqrt{r}\leq T\kappa/2$. Using \cref{lemma:upper_bound} and \cref{lemma:concentration}, it follows
    \begin{equation}
        U\leq C_1T^4\frac{\kappa^2}{4}v_{\kappa/2}^2+C_8\sqrt{r}T^3\frac{\kappa}{2} v_{\kappa/2}^2<  C_1T^4\kappa^2v_\kappa^2+C_8\sqrt{r}T^3\kappa v_\kappa^2,
    \end{equation}
    where we have used that $\Psi(\Delta)$ and $\Omega(\Delta)$ are non-decreasing, so $\Psi(\Delta)\leq \Psi(\kappa/2)$ and $\Omega(\Delta)\leq \Omega(\kappa/2)$ for all $\Delta\leq \kappa/2$. Furthermore, using \cref{cor:lower_bound} and \cref{lemma:concentration}, we have
    \begin{equation}
        L_\mathrm{near}\geq \min_{\theta\in G_K:A\kappa\leq \Delta\leq 1/T}l(\Delta),\qquad l(\Delta):=C_2T^4\Delta^2v_\Delta^2-C_8\sqrt{r}T^3\Delta v^2_\Delta.
    \end{equation}
    We want $l(\Delta)$ to be non-decreasing on $[A\kappa, 1/T]$. The first term is increasing in $\Delta$, whereas the second term is decreasing in $\Delta$. Hence, for $\Delta\geq\bar\lambda$, we require
    \begin{equation}
        \frac{\partial l(\Delta)}{\partial \Delta}\bigg|_{\Delta\geq \lambda}=4C_2T^4\Delta^3-3C_8\sqrt{r}T^3\Delta^2\geq 0\Longleftrightarrow \Delta\geq \frac{3C_8}{4C_2T}\sqrt{r},
    \end{equation}
    whereas for $\Delta \leq \bar\lambda$, we require
    \begin{equation}
        \frac{\partial l(\Delta)}{\partial \Delta}\bigg|_{\Delta\leq \lambda}=2C_2T^4\Delta\bar\lambda^2-C_8\sqrt{r}T^3\bar\lambda^2\geq 0\Longleftrightarrow \Delta\geq \frac{C_8}{2C_2T}\sqrt{r}.
    \end{equation}
    Since $\sqrt{r}\leq T\kappa/2$, a sufficient condition for both regimes is $\Delta \geq\frac{3C_8}{8C_2}\kappa$. Therefore, we set $A=\frac{3C_8}{8C_2}$. Since $b \leq 1$ and $\sigma \geq 1$, we have $C_2 = \frac{2b}{1215\pi^4} < 1.7\cdot 10^{-5}$ and $C_8 \geq 24\sqrt{32}\,\sigma^2 > 135$, so that $A > 3\cdot 10^{6} \geq 1$. Since we require $A\kappa\leq 1/T$, we must set $\kappa\leq \frac{1}{AT}$ and then require $\sqrt{r}\leq \frac{1}{2A}$. For this choice of $A$, we have
    \begin{equation}
        \min_{\theta\in G_K:A\kappa\leq \Delta\leq 1/T}l(\Delta)\geq l(A\kappa).
    \end{equation}
    Therefore, $U< L_\mathrm{near}$ is implied by
    \begin{align}
        C_1T^4\kappa^2v^2_{\kappa}+C_8\sqrt{r}T^3\kappa v^2_\kappa\leq C_2T^4A^2\kappa^2v_{A\kappa}^2-C_8\sqrt{r}T^3A\kappa v_{A\kappa}^2.
    \end{align}
    Using $v_\kappa\leq v_{A\kappa}$, a sufficient condition is
        \begin{align}
            C_1T^4\kappa^2v^2_{A\kappa}+C_8\sqrt{r}T^3\kappa v^2_{A\kappa}&\leq C_2T^4A^2\kappa^2v_{A\kappa}^2-C_8\sqrt{r}T^3A\kappa v_{A\kappa}^2\\
            (1+A)C_8\sqrt{r}T^3\kappa v^2_{A\kappa}&\leq (C_2A^2-C_1)T^4\kappa^2v_{A\kappa}^2\\
            \frac{T}{\sqrt{r}}&\geq \frac{(1+A)C_8}{C_2A^2-C_1}\frac{1}{\kappa}=:C_9\frac{1}{\kappa}.
        \end{align}
        Notably, we have
        \begin{equation}
            C_2A^2-C_1=\frac{9C_8^2}{64C_2}-C_1\geq\left( \frac{9}{64}\frac{1215\pi^4}{2b}-1\right)C_1>0,
        \end{equation}
        where we have used that $C_8^2\geq C_1$, which is true, because $C_8\geq \max(1, C_1)$, and that $b\leq 1$. Therefore, choosing 
    \begin{equation}
        \frac{T}{\sqrt{r}}\geq C_9\frac{1}{\kappa}
    \end{equation}
    is sufficient to satisfy $U< L_\mathrm{near}$. Notably, $C_9\geq \frac{C_8}{C_2A}\geq \frac{8}{3}>2$. Therefore, this choice also implies the assumption $T/\sqrt{r}\geq \frac{2}{\kappa}$. Additionally, we require $U< L_\mathrm{far}$. For $\Delta\geq 1/T$, we have
    \begin{equation}
        \Psi(\Delta)=1,\qquad \Omega(\Delta) =1.
    \end{equation}
    Hence, using \cref{cor:lower_bound} and \cref{lemma:concentration}, we have 
    \begin{equation}
        L_\mathrm{far}\geq C_2-C_8\sqrt{r},
    \end{equation}
    where we have used that $\sqrt{r}\leq T\kappa/2\leq 1$. Therefore, $U< L_\mathrm{far}$ is implied by
    \begin{equation}
        C_1T^4\kappa^2 v^2_\kappa +C_8\sqrt{r}T^3\kappa v_\kappa^2=T^3\kappa v_\kappa^2(T\kappa C_1+C_8\sqrt{r})\leq C_2-C_8\sqrt{r}.
    \end{equation}
    Since $Tv_\kappa \leq 1$ and $T\kappa \leq \frac{1}{A}$, a sufficient condition is
    \begin{equation}
        \frac{C_1}{A^2}+\frac{C_8\sqrt{r}}{A}\leq C_2-C_8\sqrt{r}\Longleftrightarrow \frac{1}{\sqrt{r}}\geq AC_9.
    \end{equation}
    Notably, since $A\kappa\leq 1/T\Leftrightarrow A\leq \frac{1}{T\kappa}$, this condition is satisfied whenever $T/\sqrt{r}\geq C_9/\kappa$. Therefore, $U< L_\mathrm{far}$ is implied by $U< L_\mathrm{near}$. Hence, choosing $T/\sqrt{r}\geq C_9/\kappa$ is sufficient to satisfy \eqref{eq:localization}.

    \item Let $\bar\lambda\geq \frac{1}{T}$. It follows $v=\bar\lambda$, $u=\Delta$ and thus
    \begin{equation}
        \Psi(\Delta)=T\Delta,\qquad \Omega(\Delta)=T\Delta.
    \end{equation}
    Furthermore, assume that $\sqrt{r}\leq T\kappa/2 $, which implies $\sqrt{r}\geq r$. Using \cref{lemma:upper_bound} and \cref{lemma:concentration}, it follows
    \begin{equation}
        U< C_1T^2\kappa^2 + C_8\sqrt{r}T\kappa,
    \end{equation}
    and, using \cref{cor:lower_bound} and \cref{lemma:concentration}, 
    \begin{equation}
        L_\mathrm{near}\geq \min_{\theta\in G_K:A\kappa\leq \Delta\leq 1/T}l(\Delta),\qquad l(\Delta):=C_2T^2\Delta^2-C_8\sqrt{r}T\Delta.
    \end{equation}
    We want $l(\Delta)$ to be non-decreasing on $[A\kappa, 1/T]$. Therefore, we require
    \begin{equation}
        \frac{\partial l(\Delta)}{\partial \Delta}=2C_2T^2\Delta-C_8\sqrt{r}T\geq 0\Longleftrightarrow \Delta\geq \frac{C_8}{2C_2T}\sqrt{r}.
    \end{equation}
    Since $\sqrt{r}\leq T\kappa/2 $, $\Delta\geq \frac{3C_8}{8C_2}\kappa$ is a sufficient condition. Therefore, we can again set $A=\frac{3C_8}{8C_2}$, and---since we require $A\kappa \leq \frac{1}{T}$---we must set $\kappa\leq\frac{1}{AT}$ and then require $\sqrt{r}\leq \frac{1}{2A}$. We then have 
    \begin{equation}
        \min_{\theta\in G_K:A\kappa\leq \Delta\leq 1/T}l(\Delta)\geq l(A\kappa).
    \end{equation}
    Therefore, $U< L_\mathrm{near}$ is implied by
    \begin{align}
        C_1T^2\kappa^2+C_8\sqrt{r}T\kappa&\leq C_2T^2A^2\kappa^2-C_8\sqrt{r}TA\kappa\\
        (1+A)C_8\sqrt{r}T\kappa &\leq (C_2A^2-C_1)T^2\kappa^2\\
        \frac{T}{\sqrt{r}}&\geq C_9\frac{1}{\kappa}.
    \end{align}
    Therefore, also in this regime, choosing $\frac{T}{\sqrt{r}}\geq C_{9}\frac{1}{\kappa}$ is sufficient to satisfy $U< L_\mathrm{near}$ and implies $\sqrt{r}\leq T\kappa/2$. Additionally, we require $U< L_\mathrm{far}$. As before, also in this regime, we have $L_\mathrm{far}\geq C_2-C_8\sqrt{r}$. Therefore, $U< L_\mathrm{far}$ is implied by
    \begin{equation}
        C_1T^2\kappa^2+C_8\sqrt{r}T\kappa =T\kappa(T\kappa C_1+C_8\sqrt{r})\leq C_2-C_8\sqrt{r}.
    \end{equation} 
    Since $T\kappa \leq \frac{1}{A}$, a sufficient condition is
    \begin{equation}
        \frac{C_1}{A^2}+\frac{C_8\sqrt{r}}{A}\leq C_2-C_8\sqrt{r}\Longleftrightarrow \frac{1}{\sqrt{r}}\geq AC_9,
    \end{equation}
    which again is implied by $\frac{T}{\sqrt{r}}=C_{9}\frac{1}{\kappa}$.
    Therefore, choosing $T/\sqrt{r}\geq C_9/\kappa$ is sufficient to satisfy \eqref{eq:localization}. 
\end{enumerate}
Consequently, choosing $T/\sqrt{r}=C_9/\kappa$ is sufficient to satisfy \eqref{eq:localization} for all $\bar\lambda$ and consequently also for all $\lambda$. Therefore, for any $\epsilon\in (0,1/T]$ and any $\delta\in (0,1)$, choosing $N\in \mathbb{N}_{\geq 1}$ such that 
\begin{equation}
    T\sqrt{N}\geq \frac{2AC_9}{\epsilon}\sqrt{\log\left(\frac{2A\pi}{\epsilon \delta} \right)} \overset{*}{\implies} T\sqrt{N}\geq \frac{2KC_9}{\pi}\sqrt{\log\left(\frac{2K}{\delta} \right)}\Longleftrightarrow \frac{T}{\sqrt{r}}\geq \frac{C_9}{\kappa}
\end{equation}
is sufficient to ensure that any minimizer $\hat{\theta}\in\arg\min_{\theta\in G_K}L^{(\lambda)}_N(\theta)$ satisfies $|\hat{\theta}-\lambda|\leq \epsilon$ with probability at least $1-\delta$. As before, $G_K$ is the uniform grid with $K=\lceil {\frac{A\pi}{2\epsilon}}\rceil \leq \frac{A\pi}{\epsilon}$, which gives $\frac{1}{\kappa}=\frac{2K}{\pi}\leq \frac{2A}{\epsilon}$ and $2K\leq \frac{2A\pi}{\epsilon}$ and therefore justifies ($*$). Furthermore, we have $\kappa=\frac{\pi}{2K}\leq \frac{\epsilon}{A}\implies A\kappa\leq \epsilon\leq 1/T$. Therefore, our previously made assumption $A\kappa\leq 1/T$ was valid. 
\end{proof}

\begin{proof}[Proof of \cref{theorem:main}]
    For any window $p$ and any $N\in \mathbb{N}_{\geq1}$, the first step of \cref{alg:ldae} is to use a quantum computer to collect $N$ i.i.d.\ draws of $(m_i, Z_i)$ to construct the loss function
    \begin{equation}
        L(\theta)=\frac{1}{N}\sum^N_{i=1}\left(Z_i-\cos(2m_i\theta)\right)^2,
    \end{equation}
    which is identical to $L^{(\lambda)}_N(\theta)$ from \cref{def:empirical_loss}, since $Z_i$ is sampled from a distribution equivalent to $p_{m_i}^{(\lambda)}$. Since $\sigma \geq 1$ and $M \geq 1$, we have $T = \max(1, M/\sigma) \leq M$ and therefore $\epsilon \leq 1/M \leq 1/T$. Moreover, by \cref{remark:definitions} every $(b,\sigma)$-admissible window is $(b,\sigma,\sigma^2,\sigma^4,\sigma^6,\sigma^8)$-admissible, so \cref{lemma:main} applies.
    Since $T \geq M/\sigma$, we have $\sigma^2T^2N \geq M^2N$, and hence choosing $N$ such that
    \begin{equation}
        M^2N \;\geq\; \frac{4\sigma^2A^2C_9^2}{\epsilon^2}\log\left(\frac{2A\pi}{\epsilon\delta}\right)
    \end{equation}
    implies \cref{eq:lemma_cond} such that \cref{lemma:main} ensures that the minimizer $\hat{\theta}$, found by \cref{alg:ldae} through minimization of $L(\theta)$ over the uniform grid $G_K$ with $K=\lceil \frac{A\pi}{2\epsilon}\rceil$, satisfies $|\hat{\theta}-\lambda|\leq \epsilon$ with probability at least $1-\delta$. In particular, choosing $N$ such that 
    \begin{equation}
        N=\left\lceil \frac{4\sigma^2A^2C_9^2}{M^2\epsilon^2}\log\left(\frac{2A\pi}{\epsilon \delta} \right)\right\rceil
    \end{equation}
    is sufficient and implies $M^2N\in {\mathcal{O}}\left( \epsilon^{-2}\log (\epsilon^{-1}\delta^{-1})\right)$. Since $A$ and $C_9$ only depend on $b$, $\sigma$, and $C, C', C'', C'''$, which can all be bounded by powers of $\sigma$, we can conclude that the hidden constant factor only depends on $b$ and $\sigma$. Furthermore, with $\hat{\theta}=\lambda+e$ with $|e|\leq \epsilon$, we can convert the guarantee on $\hat{\theta}$ to a guarantee on the estimate returned by \cref{alg:ldae}, i.e. $\hat{a}=\sin^2(\hat{\theta})$ via
    \begin{align}
        |\hat{a}-a|&=|\sin^2(\lambda+e)-\sin^2\lambda|\\&=|\sin(e)\sin(2\lambda+e)|\\&=|\sin(e)\left(\sin(2\lambda)\cos(e)+\cos(2\lambda)\sin(e) \right)|\\
        &\leq |\sin(e)||\sin(2\lambda)\cos(e)|+\sin^2(e)|\cos(2\lambda)|\\
        &\leq \sin(2\lambda)|e|+e^2\\
        &\leq 2\sqrt{a(1-a)}\epsilon+\epsilon^2,
    \end{align}
    where in the last step we have used that 
    \begin{equation}
        \sin(2\lambda)=2\sin\lambda\cos\lambda = 2\sqrt{a}\sqrt{1-a}.
    \end{equation}
\end{proof}

\begin{proof}[Proof of \cref{cor:interpolation}]
    Define the uniform window $u_M:W_M\to[0,1]$ via $u_M(m)=1/(2M+1)$. Since $T=\max(1, M/\sigma)=M$ for $\sigma=1$, and 
    \begin{equation}
        u_M(m)+u_M(-m)=\frac{2}{2M+1}\geq \frac{2}{3M}=\frac{2/3}{T}
    \end{equation}
    for every $M\in\mathbb{N}_{\geq 1}$, the uniform window $u_M$ is $(2/3,1)$-admissible for every $M\in\mathbb{N}_{\geq 1}$. 

    Set $M:=\lfloor \epsilon^{-1+\beta}\rfloor$, which satisfies $M\in \mathbb{N}_{\geq 1}$ since $\epsilon\leq 1$ and $\beta\leq 1$. We then have $M\leq \epsilon^{-1+\beta}\leq \epsilon^{-1}\implies 1/M\geq \epsilon$. By \cref{theorem:main}, applied to $u_M$ with $K = \lceil A\pi/(2\epsilon)\rceil$, there exists a choice of $N$ for which $M^2N\in{\mathcal{O}}(\epsilon^{-2}\log(\epsilon^{-1}\delta^{-1}))$ and for which \cref{alg:ldae} returns, with probability at least $1-\delta$, an estimate $\hat{\theta}$ satisfying $|\hat{\theta}-\lambda|\leq \epsilon$ for every $\lambda\in[0,\pi/2]$. Consequently, $\hat{a}=\sin^2(\hat{\theta})$ satisfies \cref{eq:error_prop} for every $a\in[0,1]$. Since $\frac{1}{2}\epsilon^{-1+\beta}\leq M\leq \epsilon^{-1+\beta}$ (because $\lfloor x \rfloor \geq x/2$ for all $x\geq 1$), this is implied by 
    \begin{equation}
        N\in{\mathcal{O}}(\epsilon^{-2}M^{-2}\log(\epsilon^{-1}\delta^{-1}))={\mathcal{O}}(\epsilon^{-2\beta}\log(\epsilon^{-1}\delta^{-1})).
    \end{equation}
    The query count is upper bounded by $MN$, so \cref{alg:ldae} requires at most
    \begin{equation}
        MN\in{\mathcal{O}}(\epsilon^{-1+\beta}\epsilon^{-2\beta}\log(\epsilon^{-1}\delta^{-1}))={\mathcal{O}}(\epsilon^{-1-\beta}\log(\epsilon^{-1}\delta^{-1}))
    \end{equation}
    queries to $U$ and $U^\dagger$, while the maximum circuit depth is $M\leq \epsilon^{-1+\beta}\in \mathcal{O}(\epsilon^{-1+\beta})$. Finally, every circuit executed by \cref{alg:ldae} consists of powers of $G$ and, for even depths, one additional reflection around $I-2P$, applied to $\ket{\psi}$, followed by a measurement of $I-2P$ or $2\ket{\psi}\bra{\psi}-I$. It therefore requires neither ancilla qubits nor controlled Grover operations.  
\end{proof}

\end{document}